\documentclass[11pt,a4paper]{article}
\usepackage{amsmath, amsfonts, amssymb}
\usepackage{a4wide}
\usepackage{parskip}
\usepackage{enumitem}
\usepackage{xcolor}

\usepackage{hyperref}
\usepackage{booktabs}
\usepackage{caption}
\usepackage{siunitx}
\usepackage{tabularx}
\usepackage{comment}
\usepackage{multirow}
\usepackage{adjustbox}
\usepackage{mathtools}
\usepackage{caption}
\usepackage{parskip}
\usepackage{graphicx}
\usepackage{adjustbox}
\usepackage{multirow}
\usepackage{amsthm}
\usepackage{bm, makecell}
\usepackage{lscape}
\usepackage{rotating}
\usepackage{floatrow}
\usepackage[noadjust]{cite}

\newtheorem{theorem}{Theorem}[section]
\newtheorem{lemma}[theorem]{Lemma}
\newtheorem{corollary}[theorem]{Corollary}
\newtheorem{proposition}[theorem]{Proposition}
\theoremstyle{definition}

\newtheorem{example}[theorem]{Example}
\newtheorem{remark}[theorem]{Remark}

\numberwithin{equation}{section}

\newcommand{\Min}{\textnormal{min}}
\newcommand{\Tor}{\textnormal{Tor}}

\usepackage{tikz,xcolor,hyperref}
\usepackage{mathdots}
\definecolor{lime}{HTML}{A6CE39}
\DeclareRobustCommand{\orcidicon}{%
	\begin{tikzpicture}
		\draw[lime, fill=lime] (0,0) 
		circle [radius=0.16] 
		node[white] {{\fontfamily{qag}\selectfont \tiny ID}};
		\draw[white, fill=white] (-0.0625,0.095) 
		circle [radius=0.007];
	\end{tikzpicture}
	\hspace{-2mm}
}

\foreach \x in {A, ..., Z}{%
	\expandafter\xdef\csname orcid\x\endcsname{\noexpand\href{https://orcid.org/\csname orcidauthor\x\endcsname}{\noexpand\orcidicon}}
}

\begin{document}
	\date{}
		\title{Structure of Polycyclic Codes over $\frac{\mathbb{F}_{p^m}[u]}{\langle u^t \rangle}$ and their Cardinalities}
		\author{{\bf Akanksha Tiwari  \footnote{email: {\tt akankshafzd8@gmail.com} }\orcidA{}, \bf Pramod Kanwar\footnote{ 	email: {\tt kanwar@ohio.edu}},  and \bf Ritumoni Sarma\footnote{	email: {\tt ritumoni407@gmail.com}}\orcidC{}} \\\\ $*,\ddagger$ Department of Mathematics \\ Indian Institute of Technology Delhi\\Hauz Khas, New Delhi-110016, India\\\\ $\dagger$ Department of Mathematics \\ Ohio University-Zanesville \\Zanesville, Ohio, U.S.A.}
  
\maketitle

\begin{abstract}
The purpose of this article is to study polycyclic codes over the ring $\frac{\mathbb{F}_{p^m}[u]}{\langle u^t \rangle}, \,t \geq 1$, and their associated torsion codes. It is shown that if $\phi$ is  a surjective ring homomorphism from a commutative ring $A$ to a Noetherian ring $B$ with $ ker(\phi)=\langle \pi\rangle$ then for every ideal $I$ of $A$, there exists $a_1,a_2,\dots,a_n$ in $I$ such that
$I=\langle a_1,a_2,\dots,a_n\rangle+\pi(I:\pi)$. Using this, we obtain generators of all ideals of the ring $\frac{\frac{\mathbb{F}_{p^m}[u]}{\langle u^t \rangle}[x]}{\langle \omega(x)\rangle},$ where $\omega(x)\in \frac{\mathbb{F}_{p^m}[u]}{\langle u^t \rangle}[x] $. For the case when $\omega(x)=f(x)^{p^s}$ where $f(x)$ is an irreducible polynomial in $\mathbb{F}_{p^m}[x]$ and $s$ is a non-negative integer,  we obtain several other results, including computation of torsion ideals and their torsional degrees when $t=4$. We use the torsional degree to compute the cardinality of polycyclic codes over the ring $\frac{\mathbb{F}_{p^m}[u]}{\langle u^4 \rangle}$ and illustrate the result with some examples that verify the computed cardinality.
\medskip


			
\medskip
			

\end{abstract}

\section{Introduction}\label{Section 1}\label{sec1}
In the early development of coding theory, the emphasis was on the study of linear codes over fields. The discussion of $\mathbb{Z}_4$ linearity of certain binary nonlinear codes including Kerdock, Preparata, Goethals codes by Hammons et al. (\cite{hammons1994z}) drew the attention of researchers working in the field of coding theory to explore linear codes initially over the ring of integers modulo $p^n$ where $p$ is a prime and then over different classes of rings (see \cite{dertli2016linear}, \cite{honold2000linear}, \cite{yildiz2007weights}, and \cite{yildiz2010linear}).  
\par
A class of linear codes, namely, constacyclic codes, is an important class of linear codes in the theory of error-correcting codes. Due to their rich algebraic structure for efficient error detection and correction, and their applications, researchers have explored this class of codes, in general, and cyclic codes, in particular, over fields as well as rings. Calderbank and Sloane (\cite{calderbank1995modular}) studied modular and $p$-adic cyclic codes and provided the structure of cyclic codes over $\mathbb{Z}_{p^m}$. Kanwar and Lopez-Permouth (\cite{kanwar1997cyclic}) independently studied cyclic codes over $\mathbb{Z}_{p^m}$ and provided a different set of generators of these codes. Wan (\cite{MR1809649}), using the techniques in \cite{kanwar1997cyclic}, generalized the results of Kanwar and Lopez-Permlouth to cyclic codes over Galois rings. Norton and Sǎlǎgean (\cite{norton2000structure}), using a different approach, generalized the structure given in \cite{calderbank1995modular} and \cite{kanwar1997cyclic} to cyclic codes over finite chain rings. These explorations of cyclic codes over different finite rings generated interest of researchers in exploring the structure of other classes of linear codes including negacyclic codes and constacyclic codes over finite rings (see for example \cite{abualrub2004mass},  \cite{dinh2005negacyclic},  \cite{Dinhdist}, \cite{dinh2008linear}, \cite{dinh2004cyclic}, \cite{taher2003generators},  \cite{wolfman1999negacyclic}).  
\par
Constacyclic codes are studied specifically in two different directions - first, when the length of the code is not divisible by the characteristic of the residue field and second, when the length of the code is divisible by the characteristic of the residue field. The codes in the second case are called repeated-root constacyclic codes, first studied by Castagnoli et al. (\cite{castagnoli1991repeated}) and van Lint (\cite{van1991repeated}). 
\par
Polycyclic codes extend the class of constacyclic codes and have algebraic properties similar to those of cyclic codes. This similarity as well as comprehensive structural description of the duals of this class of codes by Lopez-Permouth, Parra-Avila and Szabo in \cite{lopez2009dual} generated researchers' interest in exploring this class of codes over finite Galois rings, finite chain rings, etc. (see, for example, \cite{AlexandreFotue-Tabue2020AdvancesinMathematicsofCommunications} and \cite{lopez2013polycyclic}).
\par
In 2009, Dinh (\cite{dinh2009constacyclic}) obtained the structure of certain constacyclic codes of length $2^s$ over the Galois extension rings of $\mathbb{F}_2+u\mathbb{F}_2$ and also established a formula for the number of such codes. In 2010, Dinh (\cite{dinh2010constacyclic}), continuing explorations in this direction, gave the structure of constacyclic codes of length $p^s$ over $\mathbb{F}_{p^m}+u\mathbb{F}_{p^m}.$ Since then, researchers have studied the structure of cyclic and constacyclic codes of various lengths over the rings $\mathbb{F}_{p^m}$ and  $\mathbb{F}_{p^m}+u\mathbb{F}_{p^m}$ (for example, see \cite{chen2016constacyclic},
 \cite{dinh2012repeated},
and \cite{liu2016repeated}). Liu and Xu (\cite{liu2014some}) studied some constacyclic codes over $\mathbb{F}_{p^m}+u\mathbb{F}_{p^m}+u^2\mathbb{F}_{p^m}$ and Laaouine et al. (\cite{laaouine2021complete}) gave complete classification of all constacyclic codes of $p$-power length over $\mathbb{F}_{p^m}+u\mathbb{F}_{p^m}+u^2\mathbb{F}_{p^m}$ and also obtained the cardinality of these codes. Recently, Hesari and Samei (\cite{hesari2024torsion}) modified the cardinality results given in \cite{laaouine2021complete}. 
\par
In this article, continuing in the same direction of research, we study polycyclic codes over the ring $R^t:=\frac{\mathbb{F}_{p^m}[u]}{\langle u^t \rangle} = \mathbb{F}_{p^m} +u\mathbb{F}_{p^m} + \dots + u^{t-1}\mathbb{F}_{p^m}$. We first give the structure of ideals of $R^{t,\omega}:=\frac{\frac{\mathbb{F}_{p^m}[u]}{\langle u^t\rangle }[x]}{\langle \omega(x)\rangle}$ for an arbitrary polynomial $\omega(x)\in R^t[x]$ (Theorem \ref{Maintheorem1}). In particular, if $f(x)$ is an irreducible polynomial of degree $d$ over $\mathbb{F}_{p^m}$ and $\omega(x)=f(x)^{p^s}$, where $s$ is a non-negative integer, we obtain generators of all ideals of $R^{t, \omega}$ (Theorem \ref{Maintheorem}, Theorem \ref{Finalformofideals}), generalizing the results in \cite{dinh2010constacyclic} and \cite{laaouine2021complete}. We show that the ring $R^{t,\omega}$ has $2^t$ different types of ideals. As a particular case, we give the constacyclic codes of length $p^s$ over $R^t$. To obtain the structure, we first prove that if $\phi$ is  a surjective ring homomorphism from a commutative ring $A$ to a Noetherian ring $B$ with $ ker(\phi)=\langle \pi\rangle$ then for every ideal $I$ of $A$ there exist $a_1,a_2,\dots,a_n$ in $I$ such that
$I=\langle a_1,a_2,\dots,a_n\rangle+\pi(I:\pi)$ (Proposition \ref{FormofIdeals}), generalizing an earlier result where $B$ is a principal ideal ring. This ring theoretic result, which is also of independent interest, plays a crucial role in our explorations. 
  We also give the torsion codes and cardinality of polycyclic codes over $R^t$ when $t=4$ (Lemma \ref{Torsions}, Theorem \ref{Cardinality}). In addition, we give some examples of polycyclic codes illustrating and verifying our cardinality results.
\par
We remark that Boudine et al. in \cite{boudine2023classification} considered a special case of Theorem \ref{Maintheorem1} in this article. We note that while the statement of their result allows the polynomial $f(x)$ to belong to a larger ring, namely, $\frac{\mathbb{F}_{p^m}[u]}{\langle u^t\rangle}[x]$, the arguments in step 2 of their proof appear to apply only when $f(x)$ is taken from a proper subring (for any $t\ne 1$), namely, $\mathbb{F}_{p^m}[x]$. This apparent inconsistency may stem from a typographical oversight or a difference in interpretation. In this article, we address this and modify it using an alternative approach. Apart from our result being more general, its proof is self‑contained, which does not require any heavy machinery.
\par 
\section{Notation and Preliminary Results}\label{sec2}
Throughout, all rings are commutative unital rings. For a prime number $p$, $\mathbb{F}_{p^m}$ denotes the finite field with $p^m$ elements. For every prime number $p$ and any non-negative integer $t$, we use $R^t$ to denote the ring $\frac{\mathbb{F}_{p^m}[u]}{\langle u^t \rangle} = \mathbb{F}_{p^m} +u\mathbb{F}_{p^m} + \dots + u^{t-1}\mathbb{F}_{p^m}$ and $R^{t,\,\omega}$ to denote the ring $\frac{R^t[x]}{\langle \omega(x)\rangle},$ where $\omega(x)=\omega_0(x)+u\omega_1(x)+\dots+u^{t-1}\omega_{t-1}(x)\in R^t[x],$ that is,  
$$R^{t,\,\omega}=\frac{\frac{\mathbb{F}_{p^m}[u]}{\langle u^t\rangle }[x]}{\langle \omega(x)\rangle}.$$
Note that $\omega(x)\,(\textnormal{mod }u)=\omega_0(x)\in \mathbb{F}_{p^m}[x].$ We use $R^{1,\,\omega_0}$ to denote the ring $\frac{\mathbb{F}_{p^m}[x]}{\langle \omega_0(x)\rangle}.$ We observe that $R^{1,\,\omega_0}$ is a principal ideal ring. \\
For any two ideals $I$ and $J$ of a ring $R$, we use $(I:J)$ to denote their ideal quotient, that is, 
$$(I:J)=\{x\in R : xJ \subset I \}.$$
If $J = \langle a \rangle$ is a principal ideal of $R$ generated by $a$ then we write $(I:a)$ instead of $(I:\langle a \rangle )$.
The following proposition, which is also of independent interest, is crucial in our explorations.
\begin{proposition}\label{FormofIdeals}
Let $A$ and $B$ be commutative rings and let $\phi: A \rightarrow B$ be a surjective ring homomorphism with $ \ker(\phi)=\langle \pi\rangle.$ If $I$ is an ideal of $A$ such that $\phi(I)$ is a finitely generated ideal then there exist $a_1,a_2,\dots,a_n$ in $I$ such that 
$$I=\langle a_1,a_2,\dots,a_n\rangle+\pi(I:\pi),$$
where $n$ is the number of generators of $\phi(I)$. In particular, if $B$ is Noetherian then for every ideal $I$ of $A$ there exists a positive integer $n$ and $a_1,a_2,\dots,a_n$ in $I$ such that
$$I=\langle a_1,a_2,\dots,a_n\rangle+\pi(I:\pi).$$
\end{proposition}
\begin{proof}Since $\phi(I)$ is finitely generated,  $\phi(I)=\langle b_1,b_2,\dots,b_n \rangle$ for some $b_1,b_2,\dots,b_n$ in $\phi(I)$. But then there exist $a_1,a_2,\dots,a_n\in I$ such that $\phi(a_i)=b_i$ for $1\leq i\leq n$. We only need to show that $I\subset \langle a_1,a_2,\dots,a_n\rangle + \pi(I:\pi)$.  If $x\in I,$ then there exist $w_1,w_2,\dots,w_n\in B$ such that $\phi(x)=b_1w_1+b_2w_2+\dots+b_nw_n$. Since $\phi$ is surjective,  for $1\leq i\leq n$,  $w_i=\phi(v_i)$ for some $v_i\in A$. Thus $\phi(x)=\phi(a_1)\phi(v_1)+\phi(a_2)\phi(v_2)+\dots+\phi(a_n)\phi(v_n)=\phi(a_1v_1+a_2v_2+\dots+a_nv_n)$. Let $y=x-\underset{i=1}{\overset{n}{\sum}}a_iv_i$. Then $y\in I$ and $\phi(y)=0$. Thus $x=\underset{i=1}{\overset{n}{\sum}}a_iv_i+y$ with $y \in \ker(\phi)\cap I.$ Hence, $I\subset \langle a_1,a_2,\dots,a_n\rangle +\ker(\phi)\cap I. $ Since, $\ker(\phi)=\pi A,$ we have $ \ker (\phi)\cap I=\pi J,$ where $J=(I:\pi)$. Thus $I\subset \langle a_1,a_2,\dots,a_n\rangle + \pi J$. 
The last statement follows from the fact that every ideal of a Noetherian ring is finitely generated.\hfill $\square$
\end{proof}
Since a principal ideal ring is Noetherian, we have the following corollary.
\begin{corollary}\label{Idealformcite}
    Let $A$ and $B$ be two commutative rings and $\phi: A \rightarrow B$ be a surjective ring homomorphism with $\ker(\phi)=\langle \pi\rangle$. If $B$ is a principal ideal ring then for every ideal $I$ of $A$ there exists $a$ in $I$ such that 
$$I=\langle a \rangle+\pi(I:\pi).$$
\end{corollary}
 We observe that the ring homomorphism from $R^{t,\,\omega}$ to $R^{t-1,\,\omega\,(\textnormal{mod }u^{t-1})}$ given by $c(x)\mapsto c(x)\,( \textnormal{mod } u^{t-1})$  and the ring homomorphism from $R^{t,\,\omega}$ to $R^{1,\,\omega_0}$ given by $c(x)\mapsto c(x)\,(\textnormal{mod }u)$ are both surjective and have kernel $\langle u^{t-1} \rangle$ and $\langle u \rangle$, respectively. We will denote these surjective ring homomorphisms by $\Phi$ and $\mu$ respectively, that is,\\
 $$\Phi(c(x))=c(x)( \textnormal{mod } u^{t-1}) \ \textnormal{ and } \ \mu(c(x))=c(x)(\textnormal{mod }u). $$
Also $\ker (\Phi)=\langle u^{t-1} \rangle$ and $ \ker (\mu)=\langle u \rangle$. Proposition \ref{FormofIdeals}, thus, gives the following corollaries.
\begin{corollary}\label{Phi}
   For every ideal $I$ of $R^{t,\,\omega}$ there exist $a_1,a_2,\dots,a_n$ in $I$ such that
$$I=\langle a_1,a_2,\dots,a_n\rangle+u^{t-1}(I:u^{t-1}).$$
\end{corollary}
\begin{corollary}\label{mu}
   For every ideal $I$ of $R^{t,\,\omega}$ there exists $a$ in $I$ such that 
$$I=\langle a \rangle+u(I:u).$$
\end{corollary}
For a finite (commutative) ring $R$, if $\lambda \in R$ is a unit then an $R$-submodule $C$ of $R^n$ is called a $\lambda$-constacyclic code of length $n$ over $R$ if whenever $(c_0, c_1, \dots,c_{n-1}) \in C$ we have $(\lambda c_{n-1}, c_0, c_1, \dots,c_{n-2}) \in C$. $C$ is called cyclic when $\lambda = 1$ and negacyclic when $\lambda = -1$. Identifying $n$-tuples with polynomials of degree $n-1$, constacyclic codes are precisely the ideals of the ring $\frac{R[x]}{\langle x^n-\lambda \rangle}$. 

In particular, for any non-zero $\lambda $ in $\mathbb{F}_{p^m}$, the constacyclic codes of length $p^s$ over $R^t$ are precisely the ideals of the ring $R^{t,\,(x^{p^s}-\lambda)}$ and cyclic codes of length $p^s$ over $R^t$ are precisely the ideals of the ring $R^{t,(x-1)^{p^s}}$. Also for the integers $s$ and $m$, there exist integers $q$ and $r$ such that $0 \le r \le m-1$ and  $s=mq+r$. Let $\lambda_0 = \lambda^{-p^{(q+1)m-s}}=\lambda^{-p^{m-r}}$. Then $\lambda_0^{p^s} =\lambda^{-p^{(q+1)m}} = \lambda^{-1}$. It can be seen that the map $\sigma: R^{t,(x-1)^{p^s}}\rightarrow R^{t,\,(x^{p^s}-\lambda)}$ given by $c(x)\mapsto c(\lambda_0x)$ is well-defined and is a ring isomorphism. Thus, for $A \subset R^{t,(x-1)^{p^s}}$, $A$ is an ideal of $R^{t,(x-1)^{p^s}}$ if and only if $\sigma(A)$ is an ideal of $R^{t,\,(x^{p^s}-\lambda)}$. Equivalently, $A$ is a cyclic code of length $p^s$ over $R^{t}$ if and only if $\sigma(A)$ is a $\lambda$-constacyclic code of length $p^s$ over $R^{t}$.

More generally, for a polynomial $\omega(x)$  over $R^t$, polycyclic codes associated with the polynomial $\omega(x)$ over $R^t$ are precisely the ideals of the ring $R^{t,\,\omega}$.
As in the literature, if $C$ is an ideal of $R^{t,\,\omega}$, we denote the $i^\textnormal{th}\, (0\leq i \leq t-1)$ torsion of $C$ by $\Tor_i(C)$, that is,
$$\Tor_i(C)=\mu(\{c(x)\in R^{t,\,\omega}: c(x)u^i\in C\}).$$
Note that $\Tor_i(C)$ is an ideal of $R^{1,\,\omega_0}$ for $0\leq i \leq t-1$ and $\Tor_0(C)$ is the residue of $C,$ which is denoted by $\textnormal{Res}(C).$ Moreover, $v\in \Tor_i (C)$ if and only if $u^i(v+uz)\in C$ for some $z\in R^{t,\,\omega}.$ Clearly, $\Tor_i(C)\subset \Tor_{i+1}(C)$ for all $0\leq i \leq t-2$. 

\section{Ideals of \texorpdfstring{$R^{t,\,\omega}$}{}}\label{sec3}
In this section, we give the structure of ideals of the ring $R^{t,\,\omega}:=\frac{R^t[x]}{\langle \omega(x)\rangle}$, where $\omega(x) =\omega_0(x)+u\omega_1(x)+\dots+u^{t-1}\omega_{t-1}(x)\in R^t[x]$. Note that $\mu(\omega(x))=\omega_0(x) \in \mathbb{F}_{p^m}[x].$ Let
\begin{align}\label{factorization}
  \omega_0(x)=v_1(x)^{n_1}v_2(x)^{n_2}\dots v_{l}(x)^{n_l},  
\end{align}
 where for $1\leq i\leq l$, $v_i(x) \in \mathbb{F}_{p^m}[x]$ and $n_i$ is a positive integer, be the factorization of $\omega_0(x)$  into irreducible polynomials over $\mathbb{F}_{p^m}$. Then every ideal of $R^{1,\,\omega_0}:=\frac{\mathbb{F}_{p^m}[x]}{\langle \omega_0(x)\rangle}$ is of the form $\langle v_1(x)^{k_1}v_2(x)^{k_2}\dots v_{l}(x)^{k_l}\rangle,$ where $0\leq k_i\leq n_i$ for $1\leq i\leq l.$ 
Before giving the structure of ideals of $R^{t,\,\omega}$, we prove the following lemma showing that it is enough to give the structure of the ideals contained in $\langle u \rangle$. 
\begin{lemma}\label{LemmatoMaintheorem}
    Any ideal of the ring $R^{t,\,\omega}$ not contained in $\langle u \rangle $ can be expressed as 
    $$\langle v_1(x)^{k_1}v_2(x)^{k_2}\dots v_{l}(x)^{k_l}+ur(x)\rangle+J,$$ 
    where $r(x)\in R^{t,\,\omega}$, $J$ is an ideal of $R^{t,\,\omega}$ contained in $\langle u\rangle$, and for $1\leq i\leq l$, 
    $0\leq k_i\leq n_i$ (not all $k_i=n_i$).
\end{lemma}
\begin{proof}
Let $I$ be an ideal of $R^{t,\,\omega}$ not contained in $\langle u \rangle  $. Then $\mu(I)$ is a nonzero ideal of $R^{1,\,\omega_0}=\frac{\mathbb{F}_{p^m}[x]}{\langle \omega_0(x)\rangle}$. Hence $\mu(I)=\langle v_1(x)^{k_1}v_2(x)^{k_2}\dots v_{l}(x)^{k_l}\rangle,$ where $0\leq k_i\leq n_i$ (not all $k_i=n_i$) for $1\leq i\leq l.$ Hence, by Corollary \ref{mu}, $I=\langle v_1(x)^{k_1}v_2(x)^{k_2}\dots v_{l}(x)^{k_l} +ur(x)\rangle+u (I:u) $ for some $r(x)\in R^{t,\,\omega}$. Note that $u(I:u)$ is an ideal of $R^{t,\,\omega}$ contained in $\langle u\rangle.$ Hence, $I=\langle v_1(x)^{k_1}v_2(x)^{k_2}\dots v_{l}(x)^{k_l} +ur(x)\rangle+ J $, where $J=u(I:u)$ is an ideal of $R^{t,\,\omega}$ contained in $\langle u\rangle$. \hfill $\square$
\end{proof}
\begin{theorem}\label{Maintheorem1}
    The ideals of the ring $R^{t,\,\omega}$ and their generators have one of the following forms.
    \begin{itemize}
        \item [(a)] Trivial ideals $\langle 0\rangle,$ $\langle 1 \rangle.$
        \item [(b)] Any generator of a non-trivial ideal contained in $\langle u\rangle $ has the form:
        $$u^{(t-1)-i}(v_1(x)^{k_{1,i}}v_2(x)^{k_{2,i}}\cdots v_l(x)^{k_{l,i_1}})-u^{(t-1)-(i-1)}g(x),$$
        for some $0\leq i \leq t-2,\, g(x)\in R^{t,\,\omega},\,$ and $  0\leq k_{j,i}\leq n_j$ (not all $k_{j,i}=n_j$) for $1\leq j\leq l.$\\
        In fact, any ideal $I$ contained in $\langle u \rangle$ has the form:
        \begin{align*}
           I=\langle& u^{(t-1)-i_1}(v_1(x)^{k_{1,i_1}}v_2(x)^{k_{2,i_1}}\cdots v_l(x)^{k_{l,i_1}})-u^{(t-1)-(i_1-1)}g_{i_1}(x),\dots, u^{(t-1)-i_n}(\\&v_1(x)^{k_{1,i_n}}
        v_2(x)^{k_{2,i_n}}\cdots v_l(x)^{k_{l,i_n}})-u^{(t-1)-(i_n-1)}g_{i_n}(x)\rangle, 
        \end{align*}
        where $0\leq i_1<i_2<\dots<i_n\leq t-2$, $g_{i_j}(x)\in R^{t,\,\omega}$ for $1\leq j\leq n.$
        \item[(c)] Any non-trivial ideal not contained in $\langle u \rangle$ has the form:
        $$\langle (v_1(x)^{k_{1,i}}v_2(x)^{k_{2,i}}\cdots v_l(x)^{k_{l,i}})+ ur(x)\rangle+I,$$ 
        where $r(x)\in R^{t,\,\omega}$ and $I$ is an ideal as described in (b).
        \end{itemize}
\end{theorem}
    \begin{proof} Part (c) follows from Lemma \ref{LemmatoMaintheorem}, so we only need to prove Part (b). Let $I$ be a nontrivial ideal of $R^{t,\omega}$ contained in $\langle u \rangle$. We will prove the result by induction on $t$. Note that the result trivially holds for $t=1$. Let $t>1$ and assume that the result holds for $t-1$.

 Since $I \subset \langle u \rangle$, $\Phi(I) \subset \langle u \rangle $.

    Case 1. Let $\Phi(I)=\langle 0 \rangle.$
    
    Then, by Corollary \ref{Phi}, $I=u^{t-1}J$ where $J=(I:u^{t-1})$. Since $I$ is non-trivial, $J\not\subset \langle u\rangle$. Note that $\mu(J)$ is a non-zero ideal of $R^{1,\omega}$. Thus $\mu(J) =\langle v_1(x)^{k_{1}}v_2(x)^{k_{2}}\cdots v_l(x)^{k_{l}} \rangle,$ for $  0\leq k_{j}\leq n_j$ for $1\leq j\leq l$ (not all $k_{j}=n_j$). Then, by Corollary \ref{Phi}, $J=\langle v_1(x)^{k_{1}}v_2(x)^{k_{2}}\cdots v_l(x)^{k_{l}}+ur(x)\rangle+u(I:u^t)$ for some $r(x)\in R^{t,\omega}$. Hence, $I=\langle u^{t-1} v_1(x)^{k_{1}}v_2(x)^{k_{2}}\cdots v_l(x)^{k_{l}}\rangle,$ where $  0\leq k_{j}\leq n_j$ for $1\leq j\leq l$  (not all $k_{j}=n_j$) and $r(x)\in R^{t,\omega}$.



    Case 2. $\Phi(I)\not=\langle 0 \rangle.$ 

    Since $\Phi(I)\subset \langle u\rangle$, by induction hypothesis,
    \begin{align*}
       \Phi(I)=\langle& u^{(t-2)-i_1}(v_1(x)^{k_{1,i_1}}v_2(x)^{k_{2,i_1}}\cdots v_l(x)^{k_{l,i_1}})-u^{(t-2)-(i_1-1)}g_{i_1}(x), u^{(t-2)-i_2}(v_1(x)^{k_{1,i_2}}\\&v_2(x)^{k_{2,i_2}}\cdots v_l(x)^{k_{l,i_2}})-u^{(t-2)-(i_2-1)}g_{i_2}(x),\dots, u^{(t-2)-i_n}(v_1(x)^{k_{1,i_n}}v_2(x)^{k_{2,i_n}}\cdots \\&v_l(x)^{k_{l,i_n}})-u^{(t-2)-(i_n-1)}g_{i_n}(x)\rangle, 
    \end{align*} where $0\leq i_1<i_2<\dots<i_n\leq t-2,\  0\leq k_{j,i_y}\leq n_j$ for $1\leq j\leq l,\,1\leq y\leq n,$ and $g_{i_j}(x)\in R^{t-1,\phi(\omega(x))}$ for $1\leq j\leq n$.
    Hence, by Corollary \ref{Phi}, we have,
    \begin{align*}
    I=\langle &u^{(t-2)-i_1}(v_1(x)^{k_{1,i_1}}v_2(x)^{k_{2,i_1}}\cdots v_l(x)^{k_{l,i_1}})-u^{(t-2)-(i_1-1)}g_{i_1}(x)+u^{t-1}q_{1}(x), u^{(t-2)-i_2}(\\&v_1(x)^{k_{1,i_2}}v_2(x)^{k_{2,i_2}}\cdots v_l(x)^{k_{l,i_2}})-u^{(t-2)-(i_2-1)}g_{i_2}(x)+u^{t-1}q_{2}(x),\dots, u^{(t-2)-i_n}(v_1(x)^{k_{1,i_n}}\\&v_2(x)^{k_{2,i_n}}\cdots v_l(x)^{k_{l,i_n}})-u^{(t-2)-(i_n-1)}g_{i_n}(x)+u^{t-1} q_n(x)\rangle +u^{t-1}J,
    \end{align*} 
    where $J=(I:u^{t-1})$,  $q_{i}(x)\in R^{1,\,\omega_0}$ for $1\leq i\leq n$, $0\leq i_1<i_2<\dots<i_n\leq t-2, 0\leq k_{j,i_y}\leq n_j$ for $1\leq j\leq l,\,1\leq y\leq n,$ and $g_{i_j}(x)\in R^{t-1,\phi(\omega(x))}$ for $1\leq j\leq n$.\\
    Equivalently, 
    \begin{align*}
    I=&\langle u^{(t-1)-(i_1+1)}(v_1(x)^{k_{1,i_1}}v_2(x)^{k_{2,i_1}}\cdots v_l(x)^{k_{l,i_1}})-u^{(t-1)-i_1}g_{i_1}(x)+u^{t-1}q_{1}(x), u^{(t-1)-(i_2+1)}(\\&v_1(x)^{k_{1,i_2}}v_2(x)^{k_{2,i_2}}\cdots v_l(x)^{k_{l,i_2}})-u^{(t-1)-i_2}g_{i_2}(x)+u^{t-1}q_{2}(x),\dots, u^{(t-1)-(i_n+1)}(v_1(x)^{k_{1,i_n}}\\&v_2(x)^{k_{2,i_n}}\cdots v_l(x)^{k_{l,i_n}})-u^{(t-1)-i_n}g_{i_n}(x)+u^{t-1} q_n(x)\rangle +u^{t-1}J,  
    \end{align*} 
    where $J=(I:u^{t-1})$,  $q_{i}(x)\in R^{1,\,\omega_0}$ for $1\leq i\leq d$, $0\leq i_1<i_2<\dots<i_n\leq t-2, 0\leq k_{j,i_y}\leq n_j$ for $1\leq j\leq l,\,1\leq y\leq n,$ and $g_{i_j}(x)\in R^{t-1,\phi(\omega(x))}$ for $1\leq j\leq n$.

    If $J\subset \langle u \rangle,$ then $u^{t-1}J=\langle 0\rangle$ and hence
    \begin{align*}
     I=&\langle u^{(t-1)-(i_1+1)}(v_1(x)^{k_{1,i_1}}v_2(x)^{k_{2,i_1}}\cdots v_l(x)^{k_{l,i_1}})-u^{(t-1)-i_1}(g_{i_1}(x)+u^{i_1}q_{1}(x)), u^{(t-1)-(i_2+1)}\\&(v_1(x)^{k_{1,i_2}}v_2(x)^{k_{2,i_2}}\cdots v_l(x)^{k_{l,i_2}})-u^{(t-1)-i_2}(g_{i_2}(x)+u^{i_2}q_{2}(x)),\dots, u^{(t-1)-(i_n+1)}(v_1(x)^{k_{1,i_n}}\\&v_2(x)^{k_{2,i_n}}\cdots v_l(x)^{k_{l,i_n}})-u^{(t-1)-i_n}(g_{i_n}(x)+u^{i_n} q_n(x))\rangle,
     \end{align*} 
     where $q_{i}(x)\in R^{1,\,\omega_0}$ for $1\leq i\leq n$, $0\leq i_1<i_2<\dots<i_n\leq t-2, 0\leq k_{j,i_y}\leq n_j$ for $1\leq j\leq l,\,1\leq y\leq n,$ and $g_{i_j}(x)\in R^{t-1,\phi(\omega(x))}$ for $1\leq j\leq n$.

     If $J\not \subset \langle u \rangle,$ then $\mu(J)$ is a non-zero ideal of $R^{1,\,\omega_0}$. Thus, $\mu(J)=\langle v_1(x)^{k_{1}}v_2(x)^{k_{2}}\dots v_l(x)^{k_{l}}\rangle$ where $0\leq k_{i} \leq n_i$ for $1\leq i\leq l$ (not all equal to $n_i$). Then, by Corollary \ref{mu}, we have $J=\langle v_1(x)^{k_{1}}v_2(x)^{k_{2}}\cdots v_l(x)^{k_{l}}+u r(x)\rangle +u(J:u)$ for some $r(x)\in R^{t,\,\omega}$. Consequently, we have
     \begin{align*}
     I=&\langle u^{(t-1)-(i_1+1)}(v_1(x)^{k_{1,i_1}}v_2(x)^{k_{2,i_1}}\cdots v_l(x)^{k_{l,i_1}})-u^{(t-1)-i_1}(g_{i_1}(x)+u^{i_1}q_{1}(x)), u^{(t-1)-(i_2+1)}\\&(v_1(x)^{k_{1,i_2}}v_2(x)^{k_{2,i_2}}\cdots v_l(x)^{k_{l,i_2}})-u^{(t-1)-i_2}(g_{i_2}(x)+u^{i_2}q_{2}(x)),\dots, u^{(t-1)-(i_n+1)}(\omega_1(x)^{k_{1,i_n}}\\&\omega_2(x)^{k_{2,i_n}}\dots \omega_l(x)^{k_{l,i_n}})-u^{(t-1)-i_n}(g_{i_n}(x)+u^{i_n} q_n(x))\rangle+u^{t-1}(\langle v_1(x)^{k_{1}}v_2(x)^{k_{2}}\cdots v_l(x)^{k_{l}}+\\&u r(x)\rangle +u(J:u)).
     \end{align*}
     Hence, 
     \begin{align*}
    I=&\langle u^{t-1}( v_1(x)^{k_{1}}v_2(x)^{k_{2}}\cdots v_l(x)^{k_{l}}), u^{(t-1)-(i_1+1)}(v_1(x)^{k_{1,i_1}}v_2(x)^{k_{2,i_1}}\cdots v_l(x)^{k_{l,i_1}})-u^{(t-1)-i_1}\\&(g_{i_1}(x)+u^{i_1}q_{1}(x)), u^{(t-1)-(i_2+1)}(v_1(x)^{k_{1,i_2}}v_2(x)^{k_{2,i_2}}\cdots v_l(x)^{k_{l,i_2}})-u^{(t-1)-i_2}(g_{i_2}(x)+u^{i_2}\\&q_{2}(x)),\dots, u^{(t-1)-(i_n+1)}(v_1(x)^{k_{1,i_n}}v_2(x)^{k_{2,i_n}}\cdots v_l(x)^{k_{l,i_n}})-u^{(t-1)-i_n}(g_{i_n}(x)+u^{i_n} q_n(x))\rangle,
    \end{align*} 
    where $0\leq i_1<i_2<\dots<i_n\leq t-2$ and $0\leq k_j\leq n_j$ for $1\leq j\leq l$ (not all equal to $n_j$). This completes the proof.\hfill $\square$
    \end{proof}
Henceforth, we assume that $\omega(x)= f(x)^{p^s}$, where $f(x)$ is an irreducible polynomial in $\mathbb{F}_{p^m}[x]$ and $s$ is a non-negative integer. Then $\omega(x) \in \mathbb{F}_{p^m}[x]$. 
We note that for $t \neq 1$, the ring $R^{t, \omega}$, in this case, is a local ring with maximal ideal $\langle f(x), u \rangle$, but it is not a chain ring. When $t=1$, the ring $R^{1,\omega},$ is a chain ring with maximal ideal $\langle f(x)\rangle$. In fact, the ideals of $R^{1, \omega}$, in this case, are precisely $\langle f(x)^i\rangle, $ where $0\leq i\leq p^s$. For this special case of $\omega(x)$, we have the following theorem.
\begin{theorem}\label{Maintheorem}
   Let $f(x)$ be an irreducible polynomial over $\mathbb{F}_{p^m}$ and let $\omega(x)=f(x)^{p^s}$ where $s$ be a non-negative integer. Then ideals of the ring $R^{t,\omega}$ and their generators precisely have one of the following forms.
    \begin{itemize}
        \item[(a)] Trivial ideals $\langle 0 \rangle,$ $\langle 1\rangle $.
        \item[(b)] Any generator of a non-trivial ideal contained in $\langle u\rangle$ has the form $u^{(t-1)-i}f(x)^{a_{i}}\\+u^{(t-1)-(i-1)}g(x)$ where $g(x)\in R^{t,\omega},\,\, 0\le i \le t-2 $, and $0 \leq a_i \leq p^s-1$.\\
        Any such ideal $I$ has the form:
        \begin{align*}  
       I = &\langle u^{(t-1)-i_1}f(x)^{a_{i_1}}+u^{(t-1)-(i_1-1)}g_{(i_1-1)}(x), u^{(t-1)-i_2}f(x)^{a_{i_2}}+u^{(t-1)-(i_2-1)}g_{(i_2-1)}(x),\\&\dots, u^{(t-1)-i_n}f(x)^{a_{i_n}}+u^{(t-1)-(i_n-1)}g_{(i_n-1)}(x)\rangle, 
        \end{align*}
        where $0\leq i_1<i_2<\dots< i_n\leq t-2$, $0\leq a_{i_1}< a_{i_2}< \dots<a_{i_n}\leq p^{s}-1$, and $g_{i_j-1}(x)\in R^{t,\omega}$ for $1\leq j \leq n.$ .
       \item[(c)] Any non-trivial ideal not contained in $\langle  u\rangle$ has the form:
       $$\langle f(x)^{\alpha}+ur(x)\rangle+I,$$ 
       where $r(x)\in R^{t,\omega}$, $I$ is an ideal of $R^{t,\omega}$ contained in $\langle u \rangle$ (description of which is given in Part (b)), and $a_{i_n} < \alpha\leq p^{s}-1.$
    \end{itemize}
    \end{theorem}
\begin{proof}
Let $I$ be an ideal of the ring $R^{t,\omega}$ contained in $\langle u\rangle.$ Since $f(x)\in\mathbb{F}_{p^m}[x]$ is an irreducible polynomial, by Theorem \ref{Maintheorem1},
\begin{align*}
    I=\langle &u^{(t-1)-i_1}f(x)^{a_{i_1}}+u^{(t-1)-(i_1-1)}g_{(i_1-1)}(x), u^{(t-1)-i_2}f(x)^{a_{i_2}}+u^{(t-1)-(i_2-1)}g_{(i_2-1)}(x),\dots, \\&u^{(t-1)-i_n}f(x)^{a_{i_n}}+u^{(t-1)-(i_n-1)}g_{(i_n-1)}(x) \rangle,
\end{align*}
where $g_{i_j-1}(x)\in R^{t,\omega}$ for $1\leq j\leq n$ and $0\leq i_1<i_2<\dots< i_n\leq t-2$. We only have to now prove the inequality $0\leq a_{i_1}< a_{i_2}< \dots<a_{i_n}\leq p^{s}-1$.
Let, if possible, for some $1\leq j,k\leq n$ and $i_{j}>i_{k}$  we have $a_{i_j}<a_{i_k}.$ Then, 
\begin{align*}
    u^{i_{j}-i_{k}}f(x)^{a_{i_k}-a_{i_j}}\{u^{(t-1)-{i_j}}f(x)^{a_{i_j}}-u^{(t-1)-({i_j}-1)}g_{i_j-1}(x)&\}=u^{(t-1)-i_k}f(x)^{a_{i_k}}-\\&u^{(t-1)-(i_k-1)}f(x)^{a_{i_k}-a_{i_j}}g_{i_j-1}(x).
\end{align*}
Thus, the generator $u^{(t-1)-i_k}f(x)^{a_{i_k}}-u^{(t-1)-(i_k-1)}g_{i_k-1}(x)$ becomes redundant. Hence, for $0\leq i_1<i_2<\dots< i_n\leq t-2$ we have $0\leq a_{i_1}< a_{i_2}< \dots<a_{i_n}\leq p^{s}-1$.\\
Similarly, one can prove the condition $a_{i_n}<\alpha$ in Part (c) for ideals not contained in $\langle u \rangle$.   \hfill$\square$
\end{proof}
\begin{remark}
 As mentioned earlier, Boudine et al. (\cite{boudine2023classification}) considered a special case of Theorem \ref{Maintheorem1}. We note that while the statement of their result allows the polynomial $f(x)$ to belong to the ring $\frac{\mathbb{F}_{p^m}[u]}{\langle u^t\rangle}[x]$, the arguments in Step 2 of their proof appear to apply only when $f(x)$ belongs to $\mathbb{F}_{p^m}[x]$, a proper subring of $\frac{\mathbb{F}_{p^m}[u]}{\langle u^t\rangle}[x]$, when $t\ne 1$. In fact, for $t\neq 1$, if we take $f(x)=x^2-u\in \frac{\mathbb{F}_{p^m}[u]}{\langle u^t\rangle}[x]$ then $f(x)$ is irreducible over $\frac{\mathbb{F}_{p^m}[u]}{\langle u^t\rangle}$, however, $f(x)\, (\textnormal{mod }u)$ is reducible over $\mathbb{F}_{p^m}$. Hence in the generators of ideals of $\frac{\frac{\mathbb{F}_{p^m}[u]}{\langle u^t\rangle}[x]}{\langle f(x)^{p^s}\rangle}$, we will have powers of $x$ instead of powers of $f(x).$    
\end{remark}
In our next theorem, we give the number of distinct types of ideals of $R^{t,\omega}.$ 
\begin{theorem} \label{NumberofDistincttypes}
  Let $f(x)$ be an irreducible polynomial over $\mathbb{F}_{p^m}$ and $\omega(x)=f(x)^{p^s}$, where $s$ is a non-negative integer. Then the ring $R^{t,\omega}$  has precisely $2^t$ distinct types of ideals.
\end{theorem}
\begin{proof}
 To prove the theorem, we write 
$$S_i = \{u^{(t-1)-(i+1)}f(x)^{a_{(i+1)}}+u^{(t-1)-i}g_{i}(x),\dots, u^{(t-1)-(t-2)}f(x)^{a_{t-2}}+u^{(t-1)-(t-3)}\\g_{t-3}(x)\},$$
for $0\leq i\leq t-3$ and $S_{t-2}=\phi$. Then for any non-trivial ideal $I$ contained in $\langle u \rangle$ there exists an $i$ such that $0 \leq i \leq t-2$ and  
$$I=\langle u^{(t-1)-i}f(x)^{a_{i}}+u^{(t-1)-(i-1)}g_{(i-1)}(x),\, S\rangle,$$
where $S \subseteq S_i$. Using the fact that the number of subsets of a set with $n$ elements is $2^n$, we conclude that the number of distinct types of ideals of $R^{t,\omega}$ contained in $\langle u \rangle $ is $2^{t-2} + 2^{t-3}+ \dots + 1$, that is, $2^{t-1} - 1$. Since any non-trivial ideal not contained in $\langle u \rangle$ has the form
$$\langle f(x)^{\alpha}+ur(x)\rangle+J,$$ 
    where $J$ is an ideal of $R^{t,\omega}$ contained in $\langle u\rangle$, it follows that the number of distinct types of ideals of  $R^{t,\omega}$ is $2^t$, where trivial ideals are counted as one type. \hfill $\square$  
\end{proof}
Henceforth, we assume that $f(x)$ is an irreducible polynomial over $\mathbb{F}_{p^m}$ and $\omega(x)=f(x)^{p^s}$ where $s$ is a non-negative integer. Before proceeding further, we observe the following.
\begin{remark}\label{Representation}
    \begin{itemize}
        \item[(i)] An arbitrary element $c(x)$ of $R^{t,\omega}$ can be uniquely written as \\
        \begin{equation}\label{generalelement}
           c(x)=\underset{j=0}{\overset{p^s-1}{\sum}}\underset{i=0}{\overset{d-1}{\sum}}c_{i,j}^{(0)}x^if(x)^j+u\underset{j=0}{\overset{p^s-1}{\sum}}\underset{i=0}{\overset{d-1}{\sum}}c_{i,j}^{(1)}x^if(x)^j+\dots+u^{t-1}\underset{j=0}{\overset{p^s-1}{\sum}}\underset{i=0}{\overset{d-1}{\sum}}c_{i,j}^{(t-1)}x^if(x)^j 
        \end{equation}
        where $c_{i,j}^{(l)}\in \mathbb{F}_{p^m}$ for $0\leq j \leq p^s-1,\,0\leq i\leq d-1,$  and $0\leq l\leq t-1.$
        \item[(ii)] Note that $f(x)^{p^s}=0$ in $R^{1,\omega}$. Since $f(x)$ is irreducible, it follows that $x$ is a unit in $R^{1,\omega}$ and hence in $R^{t,\omega}$. Thus an element $c(x)$ (mentioned in Equation \eqref{generalelement}) of $R^{t,\omega}$ is a unit if and only if $c_{i,0}^{(0)}$ is non-zero for some $i$, $0\leq i\leq d-1$.
    \end{itemize}
    \end{remark}
Using the unique representation of $g(x) \in R^{t,\omega}$ as described in Remark \ref{Representation}, we see that any generator of the form $u^{(t-1)-i}f(x)^{a_{i}}+u^{(t-1)-(i-1)}g(x)$ can be written as 
\begin{equation}\label{generatorforideal}
    \theta_i(u,f)=u^{(t-1)-i}f(x)^{a_{i}}+u^{(t-1)-(i-1)}f(x)^{t_{i-1,0}}g_{i-1,0}(x)+\dots + u^{(t-1)}f(x)^{t_{i-1,i-1}}(x)g_{i-1,i-1}(x),
\end{equation}
where each $g_{i-1,j}(x)$ is either 0 or a unit in $R^{1,\omega}$ and $p^s-1\geq a_i>t_{i-1,0}>t_{i-1,1}>\dots>t_{i-1,i-1}\geq 0$. We, thus, have the following corollary.

\begin{corollary}\label{Finalformofideals}
    Let $f(x)$ be an irreducible polynomial over $\mathbb{F}_{p^m}$ and $\omega(x)=f(x)^{p^s},$ where $s$ is a non-negative integer. Then the ideals of the ring $R^{t,\omega}$ and their generators precisely have one of the following alternate forms.
    \begin{itemize}
        \item [(a)] Trivial ideals $\langle 0 \rangle,$ $\langle 1\rangle $.
        \item[(b)] Any generator of a non-trivial ideal contained in $\langle u\rangle$ has the form $\theta_{i}(u,f),$ as given in Equation \eqref{generatorforideal}, where for $0\leq i \leq t-2,\,0\le j\le i-1, g_{i-1,j}(x)$ is either $0$ or a unit in $R^{1,\omega}$ and $0 \leq a_i \leq p^s-1$. Any such ideal $I$ has the form:
\begin{equation*}
    I=\langle \theta_{i_1}(u,f);\theta_{i_2}(u,f);\dots;\theta_{i_n}(u,f)\rangle, 
\end{equation*}
where $0\leq i_1<i_2<\dots<i_n\leq t-2$, $0\leq a_{i_1}<a_{i_2}<\dots<a_{i_n}\leq p^s-1$, $0\leq t_{(i_1-1),(i_1-1)}<t_{(i_1-1),(i_1-2)}<\dots<t_{(i_1-1),0}<a_{i_1},\, 0\leq t_{(i_2-1),(i_2-1)}<t_{(i_2-1),(i_2-2)}<\dots<t_{(i_2-1),0}<a_{i_2},\, \dots,\,$ and $0\leq t_{(i_n-1),(i_n-1)}<t_{(i_n-1),(i_n-2)}<\dots<t_{(i_n-1),0}<a_{i_n}.$       
    \item[(c)] Any non-trivial ideal not contained in $\langle  u\rangle$ has the form:
    $$\langle f(x)^{\alpha}+uf(x)^{\alpha_1}h_1(x)+u^2f(x)^{\alpha_2}h_2(x)+\dots+u^{t-1}f(x)^{\alpha_{t-1}}h_{t-1}(x)\rangle+I,$$ 
    where $r(x)\in R^{t,\omega}$, $h_{i}(x)$ is either $0$ or unit in $R^{1,\omega}$ for $1\leq i\leq t-1$, $I$ is an ideal of $R^{t,\omega}$ contained in $\langle u \rangle$ (description of which is given in Part (b)), and $a_{i_n}<\alpha\leq p^s-1$; $0\leq \alpha_{t-1}<\alpha_{t-2}<\dots<\alpha_1<\alpha\leq p^s-1.$ Moreover, if $l\geq (t-1)-i_j,$ then $\alpha_l<a_{i_j} \, \forall \, 1\leq j\leq n.$ 
    \end{itemize}
\end{corollary}
In the next theorem, we give the precise form of the sixteen types of ideals in the case when $t=4$, that is, we give the precise form of ideals of $R^{4,\omega}$, where $\omega(x)=f(x)^{p^s}$, $f(x)$ is an irreducible polynomial over $\mathbb{F}_{p^m}$, and $s$ is a non-negative integer.
\begin{theorem}\label{idealsofS'1}
    The ideals of $R^{4,\omega}$, where $f(x)$ is an irreducible polynomial over $\mathbb{F}_{p^m}$ of degree $d$,  $\omega(x)=f(x)^{p^s}$, and $s$ is a non-negative integer, have one of the following 16 types.
   
        \begin{enumerate}
        \item {\label{Type 1}} $\langle  0 \rangle ,\,\langle 1\rangle.$
        \item {\label{Type 2}}$\langle  u^3f(x)^a \rangle,$ where $0 \leq a \leq p^s-1.$
        \item {\label{Type 3}}$\langle  u^3f(x)^{a_1},\, u^2f(x)^{a_2}+u^3f(x)^t h(x)\rangle,$ where $0 \leq t<a_1 <L< a_2 \leq p^s-1$, $h(x)$ is either $0$ or a unit in $R^{1,\omega}$, and $L$ is the smallest non-negative integer such that $u^3f(x)^L\in \langle u^2f(x)^{a_2}+u^3f(x)^t h(x)\rangle$.
        \item {\label{Type 4}} $\langle  u^3f(x)^{a_1},\, uf(x)^{a_2}+u^2f(x)^{t_1}h_{1}(x)+u^3f(x)^{t_2}h_{2}(x)  \rangle, $ where $0 \leq t_2<a_1<L < a_2 \leq p^s-1$,   $t_2<t_1<M<a_2,$ each of $h_1(x)$ and $h_2(x)$ is either $0$ or a unit in $R^{1,\omega}$, and $L, M$ are the smallest integers such that $u^3f(x)^L,\,u^2f(x)^M+u^3g(x)\in \langle uf(x)^{a_2}+u^2f(x)^{t_1}h_{1}(x)+u^3f(x)^{t_2}h_{2}(x) \rangle$ for some $g(x)\in R^{4,\omega}$.
        \item {\label{Type 5}}$\langle  u^3f(x)^{a_1},\,u^2f(x)^{a_2}+u^3f(x)^{t_1}h_1(x),\, uf(x)^{a_3}+u^2f(x)^{t_2}h_{2}(x)+u^3f(x)^{t_3}h_{3}(x) \rangle, $ where $0\leq t_i< a_1<L< a_2<M< a_3 \leq p^s-1$ for $i=1,3$, $t_3<t_2<a_2,$  $h_i(x)$ is either $0$ or a unit in $R^{1,\omega}$ for $1 \le i \le 3$, and $L,M$ are the smallest non-negative integer such that $u^3f(x)^L\in \langle u^2f(x)^{a_2}+u^3f(x)^{t_1}h_1(x),\, uf(x)^{a_3}+u^2f(x)^{t_2}h_{2}(x)+u^3f(x)^{t_3}h_{3}(x) \rangle$ and $u^2f(x)^{M}+u^3g(x)\in \langle uf(x)^{a_3}+u^2f(x)^{t_2}h_{2}(x)+u^3f(x)^{t_3}h_{3}(x) \rangle$ for some $g(x)\in R^{4,\omega}$.
        \item {\label{Type 6}} $\langle u^2f(x)^a+u^3f(x)^{t}h(x) \rangle,$ where $0 \leq t< L< a\leq p^s-1$ and $h(x)$ is either $0$ or a unit in $R^{1,\omega}$, and $L$ is the smallest integer such that $u^3f(x)^L\in \langle u^2f(x)^a+u^3f(x)^{t}h(x) \rangle$.
        \item {\label{Type 7}}$\langle u^2f(x)^{a_1}+u^3f(x)^{t_1}h_1(x),\, uf(x)^{a_2}+u^2f(x)^{t_2}h_2(x)+u^3f(x)^{t_3}h_3(x)   \rangle, $ where $0\leq t_1< a_1<L< a_2\leq p^s-1$, $0\leq t_3<t_2<a_1$, $0\leq t_i<M<L$ for $i=1,3,$ $h_i(x)$ is either $0$ or a unit in $R^{1,\omega}$ for $1 \le i \le 3$, and $L,M$ are the smallest non-negative integer such that $u^2f(x)^L+u^3g(x)\in \langle uf(x)^{a_2}+u^2f(x)^{t_2}h_2(x)+u^3f(x)^{t_3}h_3(x) \rangle$ and $u^3f(x)^M\in \langle u^2f(x)^{a_1}+u^3f(x)^{t_1}h_1(x),\, uf(x)^{a_2}+u^2f(x)^{t_2}h_2(x)+u^3f(x)^{t_3}h_3(x)   \rangle$ for some $g(x)\in R^{4,\omega}.$
        \item {\label{Type 8}} $\langle uf(x)^a+u^2f(x)^{t_1}h_1(x)+u^3f(x)^{t_{2}}h_{2}(x) \rangle, $ where $0 \leq t_2<t_1<L< a \leq p^s-1,\,0\leq t_2<M<L,$ each of $h_{1}(x)$ and $h_{2}(x)$ is either $0$ or a unit in $R^{1,\omega}$, and $L,M$ are the smallest non-negative integers such that $u^2f(x)^L+u^3g(x),\, u^3f(x)^M\in \langle uf(x)^a+u^2f(x)^{t_1}h_1(x)+u^3f(x)^{t_{2}}h_{2}(x) \rangle$ for some $g(x)\in R^{4,\omega}$.
        \item {\label{Type 9}}$\langle f(x)^{b}+uf(x)^{t_1}h_1(x)+u^2f(x)^{t_2}h_2(x)+u^3f(x)^{t_3}h_3(x) \rangle, $ where $0\leq t_3<t_2<t_1<L< b \leq p^s-1,\, 0\leq t_2<M<L,\,0\leq t_3<N<M,\,h_i(x)$ is either $0$ or a unit in $R^{1,\omega}$ for $1 \le i \le 3$, and $L,M, N$ are the smallest non-negative integers such that $uf(x)^L+u^2g_1(x),\, u^2f(x)^M+u^3g_2(x),\, u^3f(x)^N \in \langle f(x)^{b}+uf(x)^{t_1}h_1(x)+u^2f(x)^{t_2}h_2(x)+u^3f(x)^{t_3}h_3(x) \rangle$ for some $g_1(x),g_2(x)\in R^{4,\omega}$.
        \item {\label{Type 10}}$\langle f(x)^{b}+uf(x)^{t_1}h_1(x)+u^2f(x)^{t_2}h_2(x)+u^3f(x)^{t_3}h_3(x), u^3f(x)^{a}  \rangle,$ where $0\leq t_3< a<L<b \leq p^s-1$, $t_3<t_2<t_1<M< b$, $0\leq t_2<N<M,\,h_i(x)$ is either $0$ or a unit in $R^{1,\omega}$ for $1 \le i \le 3$, and $L,\,M,\,N$ are the smallest non-negative integers such that $u^3f(x)^L,\,uf(x)^M+u^2g_1(x),\, u^2f(x)^N+u^3g_2(x)\in \langle f(x)^{b}+uf(x)^{t_1}h_1(x)+u^2f(x)^{t_2}h_2(x)+u^3f(x)^{t_3}h_3(x) \rangle$ for some $g_1(x),\,g_2(x)\in R^{4,\omega}$.
        \item {\label{Type 11}}$\langle f(x)^{b}+uf(x)^{t_1}h_1(x)+u^2f(x)^{t_2}h_2(x)+u^3f(x)^{t_3}h_3(x), u^3f(x)^{a_1}, u^2f(x)^{a_2}+u^3f(x)^{t_4}\\h_4(x)  \rangle,$ where $0 \leq t_i<a_1<L<a_2<M<b\leq p^s-1$ for $i=3,4$,\,$0\leq t_3<t_2<t_1< N<b$, $0\leq t_2<a_2$,  $h_i(x)$ is either $0$ or a unit in $R^{1,\omega}$ for $1 \le i \le 4$, and $L,\,M,\,N$ are the smallest non-negative integers such that $u^3f(x)^L\in \langle f(x)^{b}+uf(x)^{t_1}h_1(x)+u^2f(x)^{t_2}h_2(x)+u^3f(x)^{t_3}h_3(x), u^2f(x)^{a_2}+u^3f(x)^{t_4}h_4(x)\rangle$ and $u^2f(x)^{M}+u^3g_1(x),\, uf(x)^N+u^2g_2(x)\in\langle f(x)^{b}+uf(x)^{t_1}h_1(x)+u^2f(x)^{t_2}h_2(x)+u^3f(x)^{t_3}h_3(x)\rangle$ for some $g_1(x),\,g_2(x)\in R^{4,\omega}.$
        \item {\label{Type 12}}$\langle f(x)^{b}+uf(x)^{t_1}h_1(x)+u^2f(x)^{t_2}h_2(x)+u^3f(x)^{t_3}h_3(x), u^3f(x)^{a_1}, uf(x)^{a_2}+u^2f(x)^{t_4}h_{4}(x)\\+u^3f(x)^{t_5}h_{5}(x)  \rangle,$ where $0\leq t_i<a_1<L<a_2<M<b\leq p^s-1$ for $i=3,5$,\, $0\leq t_3<t_2<t_1< a_2$, $0\leq t_5<t_4<a_2,\,0\leq t_i<N<L$ for $i=2,4,\,h_i(x)$ is either $0$ or a unit in $R^{1,\omega}$ for $1 \le i \le 5$, and $L,\,M,\,N$ are the smallest non-negative integers such that $u^3f(x)^L,\, u^2f(x)^N+u^3g_1(x)\in  \langle f(x)^{b}+uf(x)^{t_1}h_1(x)+u^2f(x)^{t_2}h_2(x)+u^3f(x)^{t_3}h_3(x),uf(x)^{a_2}+u^2f(x)^{t_4}h_{4}(x)+u^3f(x)^{t_5}h_{5}(x)\rangle $ and $uf(x)^{M}+u^2g_2(x)\in\langle f(x)^{b}+uf(x)^{t_1}h_1(x)+u^2f(x)^{t_2}h_2(x)+u^3f(x)^{t_3}h_3(x)\rangle$ for some $g_1(x),g_2(x)\in R^{4,\omega}.$
        \item {\label{Type 13}} $\langle f(x)^{b}+uf(x)^{t_1}h_1(x)+u^2f(x)^{t_2}h_2(x)+u^3f(x)^{t_3}h_3(x), u^3f(x)^{a_1}, u^2f(x)^{a_2}+u^3f(x)^{t_4}\\h_4(x),uf(x)^{a_3}+u^2f(x)^{t_5}h_{5}(x)+u^3f(x)^{t_6}h_{6}(x) \rangle,$ where $0 \leq t_i< a_1<L<a_2<M<a_3<N<b\leq p^s-1$ for $i=3,4,6$,\, $0\leq t_3<t_2<t_1< a_3$, $0\leq t_6<t_5<a_3$, $0\leq t_i<a_2$ for $i=2,5$, $h_i(x)$ is either $0$ or a unit in $R^{1,\omega}$ for $1 \le i \le 6$, and $L,\,M,\,N$ are the smallest non-negative integers such that $u^3f(x)^L\in  \langle f(x)^{b}+uf(x)^{t_1}h_1(x)+u^2f(x)^{t_2}h_2(x)+u^3f(x)^{t_3}h_3(x),u^2f(x)^{a_2}+u^3f(x)^{t_4}h_4(x),uf(x)^{a_3}+u^2f(x)^{t_5}h_{5}(x)+u^3\\f(x)^{t_6}h_{6}(x)\rangle $, $u^2f(x)^{M}+u^3g_1(x)\in\langle f(x)^{b}+uf(x)^{t_1}h_1(x)+u^2f(x)^{t_2}h_2(x)+u^3f(x)^{t_3}\\h_3(x),\,uf(x)^{a_3}+u^2f(x)^{t_5}h_{5}(x)+u^3f(x)^{t_6}h_{6}(x)\rangle$, and $uf(x)^N+u^2 g_2(x)\in  \langle f(x)^{b}+uf(x)^{t_1}h_1(x)+u^2f(x)^{t_2}h_2(x)+u^3f(x)^{t_3}h_3(x)\rangle$ for some $g_1(x),g_2(x)\in R^{4,\omega}.$ 
        \item {\label{Type 14}} $\langle f(x)^{b}+uf(x)^{t_1}h_1(x)+u^2f(x)^{t_2}h_2(x)+u^3f(x)^{t_3}h_3(x), u^2f(x)^{a_1}+u^3f(x)^{t_4}h_4(x), uf(x)^{a_{2}}\\+u^2f(x)^{t_5}h_{5}(x)+u^{3}f(x)^{t_6}h_{6}(x) \rangle,$ where $0\leq t_i<a_1<L<a_{2}<M<b\leq p^s-1$ for $i=3,4,6$, $0\leq t_3<t_2<t_1< b$, $0\leq t_6<t_5\leq a_1$, $t_1<a_2$, $t_2<a_1,\, 0\leq t_i<N<a_1$ for $i=3,4,6$, $h_i(x)$ is either $0$ or a unit in $R^{1,\omega}$ for $1 \le i \le 6$, and $L,\,M,\,N$ are the smallest non-negative integers such that $u^2f(x)^L+u^3g_1(x)\in  \langle f(x)^{b}+uf(x)^{t_1}h_1(x)+u^2f(x)^{t_2}h_2(x)+u^3f(x)^{t_3}h_3(x),\,uf(x)^{a_{2}}+u^2f(x)^{t_5}h_{5}(x)+u^{3}f(x)^{t_6}h_{6}(x)\rangle $, $uf(x)^{M}+u^2g_2(x)\in\langle f(x)^{b}+uf(x)^{t_1}h_1(x)+u^2f(x)^{t_2}h_2(x)+u^3f(x)^{t_3}h_3(x)\rangle,$ and $u^3f(x)^N\in \langle f(x)^{b}+uf(x)^{t_1}h_1(x)+u^2f(x)^{t_2}h_2(x)+u^3f(x)^{t_3}h_3(x), u^2f(x)^{a_1}+u^3f(x)^{t_4}h_4(x), uf(x)^{a_{2}}+u^2f(x)^{t_5}h_{5}(x)+u^{3}f(x)^{t_6}h_{6}(x) \rangle$ for some $g_1(x),g_2(x)\in R^{4,\omega}.$ 
        \item {\label{Type 15}}$\langle f(x)^{b}+uf(x)^{t_1}h_1(x)+u^2f(x)^{t_2}h_2(x)+u^3f(x)^{t_3}h_3(x), u^2f(x)^a+u^3f(x)^{t_4}h_4(x) \rangle,$ where $0\leq a<L<b\leq p^s-1$, $0\leq t_3<t_2<t_1< b$, $0\leq t_i<a$ for $i=2,4,\, 0\leq t_1<M<L,\,0\leq t_3<N<a,\,$ $h_i(x)$ is either $0$ or a unit in $R^{1,\omega}$ for $1 \le i \le 4$, and $L,\,M,\,N$ are the smallest non-negative integers such that $u^2f(x)^L+u^3g_1(x),\,uf(x)^M+u^2g_2(x)\in  \langle f(x)^{b}+uf(x)^{t_1}h_1(x)+u^2f(x)^{t_2}h_2(x)+u^3f(x)^{t_3}h_3(x)\rangle$ and $u^3f(x)^N\in\langle f(x)^{b}+uf(x)^{t_1}h_1(x)+u^2f(x)^{t_2}h_2(x)+u^3f(x)^{t_3}h_3(x), u^2f(x)^a+u^3f(x)^{t_4}h_4(x) \rangle $ for some $g_1(x),g_2(x)\in R^{4,\omega}.$
        \item {\label{Type 16}}$\langle f(x)^{b}+uf(x)^{t_1}h_1(x)+u^2f(x)^{t_2}h_2(x)+u^3f(x)^{t_3}h_3(x), uf(x)^a+u^2f(x)^{t_4}h_{4}(x)+u^3f(x)^{t_5}\\h_{5}(x) \rangle,$ where $0\leq a<L<b\leq p^s-1$, $0\leq t_3<t_2<t_1< a$, $0\leq t_5<t_4<a,\, 0\leq t_i<M<a $ for $i=2,4,\,0\leq t_i<N<M$ for $i=3,5$, $h_i(x)$ is either $0$ or a unit in $R^{1,\omega}$ for $1 \le i \le 5$, and $L,\,M,\,N$ are the smallest non-negative integers such that $uf(x)^L+u^2g_1(x)\in  \langle f(x)^{b}+uf(x)^{t_1}h_1(x)+u^2f(x)^{t_2}h_2(x)+u^3f(x)^{t_3}h_3(x)\rangle,\, u^2f(x)^M+u^3g_2(x),\, u^3f(x)^N\in\langle f(x)^{b}+uf(x)^{t_1}h_1(x)+u^2f(x)^{t_2}h_2(x)+u^3f(x)^{t_3}h_3(x), uf(x)^a+u^2f(x)^{t_4}h_{4}(x)+u^3f(x)^{t_5}h_{5}(x) \rangle$ for some $g_1(x),\,g_2(x)\in R^{4,\omega}.$
    \end{enumerate}
\end{theorem}
\textbf{Remark}\label{C iff u^iC}: It is easy to see that for any ideal $C$ of $R^{4,\omega}$, $u^2f(x)^L+u^3g(x)\in C$ for some $g(x) \in R^{4,\omega}$ if and only if $u^3f(x)^L\in uC$ and $uf(x)^L+u^2g(x)\in C$ for some $g(x) \in R^{4,\omega}$ if and only if $u^3f(x)^L\in u^2C$. 
Thus, the smallest non-negative integer $L$ such that $u^3f(x)^L\in \langle u^2f(x)^{a}+u^3f(x)^{t} h(x)\rangle$, the smallest non-negative integer $L$ such that $u^2f(x)^L+u^3g(x)\in \langle uf(x)^{a}+u^2f(x)^{t}h(x)+u^3f(x)^{t_0}h_0(x) \rangle$, and the smallest non-negative integer $L$ such that $uf(x)^L+u^2g(x) \in \langle f(x)^{a}+uf(x)^{t}h(x)+u^2f(x)^{t_0}h_0(x)+u^3f(x)^{t_1}h_1(x) \rangle$  are all same. In other words, the computation of $M$ in Part (4) (as well as in Part (5)) of Theorem \ref{idealsofS'1} is same as the computation of $L$ in Part (3) of the theorem. Similarly, the computation of $M$ in Part (10) (also that of $N$ in Part (11)) of Theorem \ref{idealsofS'1} is same as the computation of $L$ in Part (3) of the theorem. 

Recall that the cyclic codes of length $n$ over a finite commutative ring $R$ are precisely the ideals of the ring $\frac{R[x]}{\langle x^n-1\rangle}$. Thus, taking $f(x)=x-1$ for the irreducible polynomial and writing $\omega(x)=f(x)^{p^s}=(x-1)^{p^s}$ we get, as a particular case of Theorem \ref{idealsofS'1}, all types of cyclic codes of length $p^s$ over the ring $R^4=\frac{\mathbb{F}_{p^m}[u]}{\langle u ^4 \rangle}$. For the sake of completeness, we list them in the following theorem. 
Also, it may be noted that if we write $u^3=0$ in the following theorem, we get generators of all ideals of the ring $R^{3,\,\omega}$ (equivalently, cyclic codes over $R^3$), as given in \cite{laaouine2021complete}.  
\begin{theorem}\label{idealsofS'}
    The ideals of the ring $R^{4,\,\omega}$, where $\omega(x)=(x-1)^{p^s}$, equivalently, the cyclic codes of length $p^s$ over the ring $R^4$, have one of the following sixteen types.
    \begin{enumerate}
       \item  $\langle  0 \rangle ,\,\langle 1\rangle.$
        \item $\langle  u^3(x-1)^a \rangle,$ where $0 \leq a \leq p^s-1.$
        \item $\langle  u^3(x-1)^{a_1},\, u^2(x-1)^{a_2}+u^3(x-1)^t h(x)\rangle,$ where $0 \leq t<a_1 <L< a_2 \leq p^s-1$, $h(x)$ is either $0$ or a unit in $R^{1,\omega}$, and $L$ is the smallest non-negative integer such that $u^3(x-1)^L\in \langle u^2(x-1)^{a_2}+u^3(x-1)^t h(x)\rangle$.
        \item $\langle  u^3(x-1)^{a_1},\, u(x-1)^{a_2}+u^2(x-1)^{t_1}h_{1}(x)+u^3(x-1)^{t_2}h_{2}(x)  \rangle, $ where $0 \leq t_2<a_1<L < a_2 \leq p^s-1$,   $t_2<t_1<M<a_2,$ each of $h_1(x)$ and $h_2(x)$ is either $0$ or a unit in $R^{1,\omega}$, and $L, M$ are the smallest non-negative integers such that $u^3(x-1)^L,\,u^2(x-1)^M+u^3g(x)\in \langle u(x-1)^{a_2}+u^2(x-1)^{t_1}h_{1}(x)+u^3(x-1)^{t_2}h_{2}(x) \rangle$ for some $g(x)\in R^{4,\omega}$.
        \item $\langle  u^3(x-1)^{a_1},\,u^2(x-1)^{a_2}+u^3(x-1)^{t_1}h_1(x),\, u(x-1)^{a_3}+u^2(x-1)^{t_2}h_{2}(x)+u^3(x-1)^{t_3}h_{3}(x) \rangle, $ where $0\leq t_i< a_1<L< a_2<M< a_3 \leq p^s-1$ for $i=1,3$, $t_3<t_2<a_2,$  $h_i(x)$ is either $0$ or a unit in $R^{1,\omega}$ for $1 \le i \le 3$, and $L,M$ are the smallest non-negative integer such that $u^3(x-1)^L\in \langle u^2(x-1)^{a_2}+u^3(x-1)^{t_1}h_1(x),\, u(x-1)^{a_3}+u^2(x-1)^{t_2}h_{2}(x)+u^3(x-1)^{t_3}h_{3}(x) \rangle$ and $u^2(x-1)^{M}+u^3g(x)\in \langle u(x-1)^{a_3}+u^2(x-1)^{t_2}h_{2}(x)+u^3(x-1)^{t_3}h_{3}(x) \rangle$ for some $g(x)\in R^{4,\omega}$.
        \item  $\langle u^2(x-1)^a+u^3(x-1)^{t}h(x) \rangle,$ where $0 \leq t< L< a\leq p^s-1$ and $h(x)$ is either $0$ or a unit in $R^{1,\omega}$, and $L$ is the smallest non-negative integer such that $u^3(x-1)^L\in \langle u^2(x-1)^a+u^3(x-1)^{t}h(x) \rangle$.
        \item $\langle u^2(x-1)^{a_1}+u^3(x-1)^{t_1}h_1(x),\, u(x-1)^{a_2}+u^2(x-1)^{t_2}h_2(x)+u^3(x-1)^{t_3}h_3(x)   \rangle, $ where $0\leq t_1< a_1<L< a_2\leq p^s-1$, $0\leq t_3<t_2<a_1$, $0\leq t_i<M<L$ for $i=1,3,$ $h_i(x)$ is either $0$ or a unit in $R^{1,\omega}$ for $1 \le i \le 3$, and $L,M$ are the smallest non-negative integer such that $u^2(x-1)^L+u^3g(x)\in \langle u(x-1)^{a_2}+u^2(x-1)^{t_2}h_2(x)+u^3(x-1)^{t_3}h_3(x) \rangle$ and $u^3(x-1)^M\in \langle u^2(x-1)^{a_1}+u^3(x-1)^{t_1}h_1(x),\, u(x-1)^{a_2}+u^2(x-1)^{t_2}h_2(x)+u^3(x-1)^{t_3}h_3(x)   \rangle$ for some $g(x)\in R^{4,\omega}.$
        \item $\langle u(x-1)^a+u^2(x-1)^{t_1}h_1(x)+u^3(x-1)^{t_{2}}h_{2}(x) \rangle, $ where $0 \leq t_2<t_1<L< a \leq p^s-1,\,0\leq t_2<M<L,$ each of $h_{1}(x)$ and $h_{2}(x)$ is either $0$ or a unit in $R^{1,\omega}$, and $L,M$ are the smallest non-negative integers such that $u^2(x-1)^L+u^3g(x),\, u^3(x-1)^M\in \langle u(x-1)^a+u^2(x-1)^{t_1}h_1(x)+u^3(x-1)^{t_{2}}h_{2}(x) \rangle$ for some $g(x)\in R^{4,\omega}$.
        \item $\langle (x-1)^{b}+u(x-1)^{t_1}h_1(x)+u^2(x-1)^{t_2}h_2(x)+u^3(x-1)^{t_3}h_3(x) \rangle, $ where $0\leq t_3<t_2<t_1<L< b \leq p^s-1,\, 0\leq t_2<M<L,\,0\leq t_3<N<M,\,h_i(x)$ is either $0$ or a unit in $R^{1,\omega}$ for $1 \le i \le 3$, and $L,\,M,\,N$ are the smallest non-negative integers such that $u(x-1)^L+u^2g_1(x),\, u^2(x-1)^M+u^3g_2(x),\, u^3(x-1)^N \in \langle (x-1)^{b}+u(x-1)^{t_1}h_1(x)+u^2(x-1)^{t_2}h_2(x)+u^3(x-1)^{t_3}h_3(x) \rangle$ for some $g_1(x),g_2(x)\in R^{4,\omega}$.
        \item $\langle (x-1)^{b}+u(x-1)^{t_1}h_1(x)+u^2(x-1)^{t_2}h_2(x)+u^3(x-1)^{t_3}h_3(x), u^3(x-1)^{a}  \rangle,$ where $0\leq t_3< a<L<b \leq p^s-1$, $t_3<t_2<t_1<M< b$, $0\leq t_2<N<M,\,h_i(x)$ is either $0$ or a unit in $R^{1,\omega}$ for $1 \le i \le 3$, and $L,\,M,\,N$ are the smallest non-negative integers such that $u^3(x-1)^L,\,u(x-1)^M+u^2g_1(x),\, u^2(x-1)^N+u^3g_2(x)\in \langle (x-1)^{b}+u(x-1)^{t_1}h_1(x)+u^2(x-1)^{t_2}h_2(x)+u^3(x-1)^{t_3}h_3(x) \rangle$ for some $g_1(x),\,g_2(x)\in R^{4,\omega}$.
        \item $\langle (x-1)^{b}+u(x-1)^{t_1}h_1(x)+u^2(x-1)^{t_2}h_2(x)+u^3(x-1)^{t_3}h_3(x), u^3(x-1)^{a_1}, u^2(x-1)^{a_2}+u^3(x-1)^{t_4}h_4(x)  \rangle,$ where $0 \leq t_i<a_1<L<a_2<M<b\leq p^s-1$ for $i=3,4$,\,$0\leq t_3<t_2<t_1< N<b$, $0\leq t_2<a_2$,  $h_i(x)$ is either $0$ or a unit in $R^{1,\omega}$ for $1 \le i \le 4$, and $L,\,M,\,N$ are the smallest non-negative integers such that $u^3(x-1)^L\in \langle (x-1)^{b}+u(x-1)^{t_1}h_1(x)+u^2(x-1)^{t_2}h_2(x)+u^3(x-1)^{t_3}h_3(x), u^2(x-1)^{a_2}+u^3(x-1)^{t_4}h_4(x)\rangle$ and $u^2(x-1)^{M}+u^3g_1(x),\, u(x-1)^N+u^2g_2(x)\in\langle (x-1)^{b}+u(x-1)^{t_1}h_1(x)+u^2(x-1)^{t_2}h_2(x)+u^3(x-1)^{t_3}h_3(x)\rangle$ for some $g_1(x),\,g_2(x)\in R^{4,\omega}.$
        \item $\langle (x-1)^{b}+u(x-1)^{t_1}h_1(x)+u^2(x-1)^{t_2}h_2(x)+u^3(x-1)^{t_3}h_3(x), u^3(x-1)^{a_1}, u(x-1)^{a_2}+u^2(x-1)^{t_4}h_{4}(x)+u^3(x-1)^{t_5}h_{5}(x)  \rangle,$ where $0\leq t_i<a_1<L<a_2<M<b\leq p^s-1$ for $i=3,5$,\, $0\leq t_3<t_2<t_1< a_2$, $0\leq t_5<t_4<a_2,\,0\leq t_i<N<L$ for $i=2,4,\,h_i(x)$ is either $0$ or a unit in $R^{1,\omega}$ for $1 \le i \le 5$, and $L,\,M,\,N$ are the smallest non-negative integers such that $u^3(x-1)^L,\, u^2(x-1)^N+u^3g_1(x)\in  \langle (x-1)^{b}+u(x-1)^{t_1}h_1(x)+u^2(x-1)^{t_2}h_2(x)+u^3(x-1)^{t_3}h_3(x),u(x-1)^{a_2}+u^2(x-1)^{t_4}h_{4}(x)+u^3(x-1)^{t_5}h_{5}(x)\rangle $ and $u(x-1)^{M}+u^2g_2(x)\in\langle (x-1)^{b}+u(x-1)^{t_1}h_1(x)+u^2(x-1)^{t_2}h_2(x)+u^3(x-1)^{t_3}h_3(x)\rangle$ for some $g_1(x),g_2(x)\in R^{4,\omega}.$
        \item $\langle (x-1)^{b}+u(x-1)^{t_1}h_1(x)+u^2(x-1)^{t_2}h_2(x)+u^3(x-1)^{t_3}h_3(x), u^3(x-1)^{a_1}, u^2(x-1)^{a_2}+u^3(x-1)^{t_4}h_4(x),u(x-1)^{a_3}+u^2(x-1)^{t_5}h_{5}(x)+u^3(x-1)^{t_6}h_{6}(x) \rangle,$ where $0 \leq t_i< a_1<L<a_2<M<a_3<N<b\leq p^s-1$ for $i=3,4,6$,\, $0\leq t_3<t_2<t_1< a_3$, $0\leq t_6<t_5<a_3$, $0\leq t_i<a_2$ for $i=2,5$, $h_i(x)$ is either $0$ or a unit in $R^{1,\omega}$ for $1 \le i \le 6$, and $L,\,M,\,N$ are the smallest non-negative integers such that $u^3(x-1)^L\in  \langle (x-1)^{b}+u(x-1)^{t_1}h_1(x)+u^2(x-1)^{t_2}h_2(x)+u^3(x-1)^{t_3}h_3(x),u^2(x-1)^{a_2}+u^3(x-1)^{t_4}h_4(x),u(x-1)^{a_3}+u^2(x-1)^{t_5}h_{5}(x)+u^3(x-1)^{t_6}h_{6}(x)\rangle $, $u^2(x-1)^{M}+u^3g_1(x)\in\langle (x-1)^{b}+u(x-1)^{t_1}h_1(x)+u^2(x-1)^{t_2}h_2(x)+u^3(x-1)^{t_3}h_3(x),\,u(x-1)^{a_3}+u^2(x-1)^{t_5}h_{5}(x)+u^3(x-1)^{t_6}h_{6}(x)\rangle$, and $u(x-1)^N+u^2 g_2(x)\in  \langle (x-1)^{b}+u(x-1)^{t_1}h_1(x)+u^2(x-1)^{t_2}\\h_2(x)+u^3(x-1)^{t_3}h_3(x)\rangle$ for some $g_1(x),g_2(x)\in R^{4,\omega}.$ 
        \item  $\langle (x-1)^{b}+u(x-1)^{t_1}h_1(x)+u^2(x-1)^{t_2}h_2(x)+u^3(x-1)^{t_3}h_3(x), u^2(x-1)^{a_1}+u^3(x-1)^{t_4}h_4(x), u(x-1)^{a_{2}}+u^2(x-1)^{t_5}h_{5}(x)+u^{3}(x-1)^{t_6}h_{6}(x) \rangle,$ where $0\leq t_i<a_1<L<a_{2}<M<b\leq p^s-1$ for $i=3,4,6$, $0\leq t_3<t_2<t_1< b$, $0\leq t_6<t_5\leq a_1$, $t_1<a_2$, $t_2<a_1,\, 0\leq t_i<N<a_1$ for $i=3,4,6$, $h_i(x)$ is either $0$ or a unit in $R^{1,\omega}$ for $1 \le i \le 6$, and $L,M,N$ are the smallest non-negative integers such that $u^2(x-1)^L+u^3g_1(x)\in  \langle (x-1)^{b}+u(x-1)^{t_1}h_1(x)+u^2(x-1)^{t_2}h_2(x)+u^3(x-1)^{t_3}h_3(x),\,u(x-1)^{a_{2}}+u^2(x-1)^{t_5}h_{5}(x)+u^{3}(x-1)^{t_6}h_{6}(x)\rangle $, $u(x-1)^{M}+u^2g_2(x)\in\langle (x-1)^{b}+u(x-1)^{t_1}h_1(x)+u^2(x-1)^{t_2}h_2(x)+u^3(x-1)^{t_3}h_3(x)\rangle,$ and $u^3(x-1)^N\in \langle (x-1)^{b}+u(x-1)^{t_1}h_1(x)+u^2(x-1)^{t_2}h_2(x)+u^3(x-1)^{t_3}h_3(x), u^2(x-1)^{a_1}+u^3(x-1)^{t_4}h_4(x), u(x-1)^{a_{2}}+u^2(x-1)^{t_5}h_{5}(x)+u^{3}\\(x-1)^{t_6}h_{6}(x) \rangle$ for some $g_1(x),g_2(x)\in R^{4,\omega}.$ 
        \item $\langle (x-1)^{b}+u(x-1)^{t_1}h_1(x)+u^2(x-1)^{t_2}h_2(x)+u^3(x-1)^{t_3}h_3(x), u^2(x-1)^a+u^3(x-1)^{t_4}h_4(x) \rangle,$ where $0\leq a<L<b\leq p^s-1$, $0\leq t_3<t_2<t_1< b$, $0\leq t_i<a$ for $i=2,4,\, 0\leq t_1<M<L,\,0\leq t_3<N<a,\,$ $h_i(x)$ is either $0$ or a unit in $R^{1,\omega}$ for $1 \le i \le 4$, and $L,\,M,\,N$ are the smallest non-negative integers such that $u^2(x-1)^L+u^3g_1(x),\,u(x-1)^M+u^2g_2(x)\in  \langle (x-1)^{b}+u(x-1)^{t_1}h_1(x)+u^2(x-1)^{t_2}h_2(x)+u^3(x-1)^{t_3}h_3(x)\rangle$ and $u^3(x-1)^N\in\langle (x-1)^{b}+u(x-1)^{t_1}h_1(x)+u^2(x-1)^{t_2}h_2(x)+u^3(x-1)^{t_3}h_3(x), u^2(x-1)^a+u^3(x-1)^{t_4}h_4(x) \rangle $ for some $g_1(x),g_2(x)\in R^{4,\omega}.$
        \item $\langle (x-1)^{b}+u(x-1)^{t_1}h_1(x)+u^2(x-1)^{t_2}h_2(x)+u^3(x-1)^{t_3}h_3(x), u(x-1)^a+u^2(x-1)^{t_4}h_{4}(x)+u^3(x-1)^{t_5}h_{5}(x) \rangle,$ where $0\leq a<L<b\leq p^s-1$, $0\leq t_3<t_2<t_1< a$, $0\leq t_5<t_4<a,\, 0\leq t_i<M<a $ for $i=2,4,\,0\leq t_i<N<M$ for $i=3,5$, $h_i(x)$ is either $0$ or a unit in $R^{1,\omega}$ for $1 \le i \le 5$, and $L,\,M,\,N$ are the smallest non-negative integers such that $u(x-1)^L+u^2g_1(x)\in  \langle (x-1)^{b}+u(x-1)^{t_1}h_1(x)+u^2(x-1)^{t_2}h_2(x)+u^3(x-1)^{t_3}h_3(x)\rangle,\, u^2(x-1)^M+u^3g_2(x),\, u^3(x-1)^N\in\langle (x-1)^{b}+u(x-1)^{t_1}h_1(x)+u^2(x-1)^{t_2}h_2(x)+u^3(x-1)^{t_3}h_3(x), u(x-1)^a+u^2(x-1)^{t_4}h_{4}(x)+u^3(x-1)^{t_5}h_{5}(x) \rangle$ for some $g_1(x),\,g_2(x)\in R^{4,\omega}.$
    \end{enumerate}
\end{theorem}
Finally, recall that for any non-zero element $\lambda $ in $\mathbb{F}_{p^m}$, using the ring isomorphism  $\sigma: R^{4,(x-1)^{p^s}}\rightarrow R^{4,(x^{p^s}-\lambda)}$ discussed in Section \ref{sec2}, we see that $C$ is a cyclic code of length $p^s$ over $R^{4}$ if and only if $\sigma(C)$ is a $\lambda$-constacyclic code of length $p^s$ over $R^{4}$. Using this association and Theorem \ref{idealsofS'}, we get the form of all $\lambda$-constacyclic codes of length $p^s$ over $R^{4}$. 
\section{Cardinality of ideals of \texorpdfstring{$R^{4,\omega}$}{}}
In this section, for an irreducible polynomial $f(x)$ over $\mathbb{F}_{p^m}$ of degree $d$ and $\omega(x)=f(x)^{p^s}$ where $s$ is a non-negative integer, we discuss torsion ideals of the ideals of $R^{4, \omega}$ and use these to obtain the cardinality of the ideals of $R^{4, \omega}$. We first compute the parameters mentioned in Theorem \ref{idealsofS'1}. In fact, these parameters will be critical in determining the cardinality of the ideals of $R^{4,\omega}$. For torsion ideals and parameters in the case when $t=3$ and $f(x)=x-1$, one can refer to Laaouine et. al. (\cite{laaouine2021complete}) and Hesari and Samei (\cite{hesari2024torsion}).
\begin{proposition}\label{proofofresult}
    Let $L$ be the smallest non-negative integer such that  $u^3f(x)^{L}\in\langle u^2f(x)^{a}+u^3f(x)^{t}h(x)\rangle,$ where $h(x)$, if non-zero, is a unit in $R^{1,\omega}$. Then, $$L= \begin{cases} 
      a & \textnormal{ if } h(x)=0, \\
      \Min\{a,\,p^s-a+t\}  & \textnormal{ if } h(x) \ne 0. \\
   \end{cases}$$
\end{proposition}
 \begin{proof}
Let $C= \langle u^2f(x)^{a}+u^3f(x)^{t}h(x)\rangle$ and let $\gamma$ be a non-negative integer such that $u^3f(x)^{\gamma}\in C$. Then there exists 
$c_0(x), c_1(x)\in R^{1,\omega}$ such that
\[u^2f(x)^{a}c_0(x)+u^3f(x)^{a}c_1(x)+u^3f(x)^{t}h(x)c_0(x)=u^3f(x)^{\gamma}\]
and hence
\begin{align}
&f(x)^a c_0(x) = 0, \tag{1}\\
&f(x)^{a}c_1(x) + f(x)^t h(x)c_0(x)=f(x)^{\gamma}. \tag{2}
\end{align}
\textbf{Case 1.} $h(x)=0.$
In this case, since $u^3f(x)^a$ is in the ideal, we have $L=a.$\\
\textbf{Case 2.} $h(x)\ne 0.$
Equation (1) gives $c_0(x)=f(x)^{p^s-a}\tilde{c}_0(x)$ for some $\tilde{c}_0(x)\in R^{1,\omega}$. Using this in Equation (2), we get\\
\[f(x)^ac_1(x)+f(x)^{p^s-a+t}h(x)\tilde{c}_0(x)=f(x)^{\gamma}.\]
Thus $\gamma \geq \min\{a,\, p^s-a+t\}.$ In particular, since $u^3f(x)^{L}\in C$, we have $L\geq \min\{a,\, p^s-a+t\}.$ Also, if we take $c_1(x)=1,\, c_0(x)=0$, we get $u^3f(x)^a\in C$ and if we take $c_1(x)=0,\, c_0(x)=f(x)^{p^s-a}h(x)^{-1}$, we get $u^3f(x)^{p^s-a+t}\in C.$ Since $L$ is the smallest non-negative integer satisfying $u^3f(x)^{L}\in C $, we have  $L\leq \min\{a,\, p^s-a+t\}$ and hence $L=\min\{a,\,p^s-a+t\}.$\hfill $\square$
\end{proof}    
\begin{proposition} \label{Parameter 2}
    Let $L$ be the smallest non-negative integer such that $u^3f(x)^{L}\in\langle uf(x)^{a}+u^2f(x)^{t_1}h_1(x)+u^3f(x)^{t_2}h_2(x)\rangle,$ where for $1 \le i \le 2$,  $h_i(x)$, if non-zero, is a unit in $R^{1,\omega}$. Then,
    
$L= \begin{cases}
a, & \text{if } h_1(x) = h_2(x) = 0, \\
\min\{a, p^s - a + t_2\}, & \text{if } h_1(x) = 0 \text{ and } h_2(x) \neq 0, \\
\min\{a, p^s - 2(a-t_1)\}, & \text{if }  h_1(x) \neq 0,\, h_2(x)=0, \text{ and } a\leq p^s-a+t_1, \\

t_1, & \text{if }  h_1(x) \neq 0,\,h_2(x)=0, \text{ and } a\geq p^s-a+t_1, \\

\min\{a,\, p^s - a+t_1,\,\beta_1\}, & \text{if }  h_1(x) \neq 0,\, h_2(x)\neq 0, \text{ and } a\leq p^s-a+t_1, \\

\min\{a,\, p^s+t_2-t_1,\,\beta_2\}, & \text{if }  h_1(x) \neq 0,\, h_2(x)\neq 0, \text{ and } a\geq p^s-a+t_1, \\
\end{cases}$

where $\beta_1=\textnormal{max}\{k\,:\,f(x)^k\mid (f(x)^{p^s-a+t_2}h_2(x)-f(x)^{p^s-2a+2t_1}h_1^2(x))\}$ and $\beta_2=\textnormal{max}\{k\,:\,f(x)^k\mid (f(x)^{t_1}h_1(x)-f(x)^{a+t_2-t_1}h_2(x)h_1^{-1}(x))\}$ .
\end{proposition}

\begin{proof}
Let $C=\langle uf(x)^{a}+u^2f(x)^{t_1}h_1(x)+u^3f(x)^{t_2}h_2(x)\rangle$ and let $\gamma$ be a non-negative integer such that $u^3f(x)^{\gamma} \in C$. Then there exists  
$c_0(x), c_1(x), c_2(x),c_3(x) \in R^{1,\omega}$ such that,
\begin{align*}
  uf(x)^a c_0(x) +u^2\bigl(f(x)^{t_1}h_1(x)c_0(x) + &f(x)^ac_1(x)\bigr)+u^3\bigl(f(x)^{t_2}h_2(x)c_0(x)\\&+f(x)^{t_1}h_1(x)c_1(x)+f(x)^ac_2(x)\bigr)=u^3f(x)^{\gamma}  
\end{align*}
and hence
\begin{align}
&f(x)^a c_0(x) = 0, \tag{1}\\
&f(x)^{t_1}h_1(x)c_0(x) + f(x)^a c_1(x) = 0, \tag{2}\\
&f(x)^{t_2}h_2(x)c_0(x) + f(x)^a c_2(x) + f(x)^{t_1}h_1(x)c_1(x)=f(x)^{\gamma}. \tag{3}
\end{align}

Equation (1) gives $c_0(x) = f(x)^{p^s-a}\tilde{c}_0(x)$, where $\tilde{c}_0(x)\in R^{1,\omega}$.  
Using this in Equations (2) and (3) we get
\begin{align}
&f(x)^{p^s-a+t_1}h_1(x)\tilde{c}_0(x) + f(x)^a c_1(x) = 0, \tag{4}\\
&f(x)^{p^s-a+t_2}h_2(x)\tilde{c}_0(x) + f(x)^a c_2(x) + f(x)^{t_1}h_1(x)c_1(x)=f(x)^{\gamma}. \tag{5}
\end{align}

\textbf{Case 1.} $h_1(x) = h_2(x) = 0$. In this case, since $u^3f(x)^a$ is in the ideal, we have $L=a$.\\
\textbf{Case 2.} $h_1(x) = 0$ and $h_2(x) \neq 0$. In this case, Equation (5) reduces to
    \[
    f(x)^{p^s-a+t_2}h_2(x)\tilde{c}_0(x) + f(x)^a c_2(x)=f(x)^{\gamma}.
    \] 
Thus $\gamma \geq \min\{a,\,p^s-a+t_2\}$. In particular, since $u^3f(x)^{L}\in C$, we have $L\geq\min\{a,\,p^s-a+t_2\}.$ Also if we take $c_1(x)=c_0(x)=0, c_2(x)=1$, we get $u^3f(x)^a \in C$ and if we take $c_0(x)=h_2(x)^{-1}, c_1(x)=c_2(x)=0$, we get $u^3f(x)^{p^s-a+t_2} \in C$. Since $L$ is the smallest non-negative integer satisfying $u^3f(x)^{L} \in C$, we have $L\leq \min\{a,p^s-a+t_2\}$ and hence $L= \min\{a, p^s-a+t_2\}$.\\
\textbf{Case 3. } $h_1(x)\neq 0$, $h_2(x)=0$. 
Equation (5), under these conditions, reduces to
$$f(x)^a c_2(x) + f(x)^{t_1}h_1(x)c_1(x)=f(x)^{\gamma}.$$
Since $a > t_1$, we have $\gamma \geq t_1$. Also, Equation (4) is
$$f(x)^{p^s-a+t_1}h_1(x)\tilde{c}_0(x) + f(x)^a c_1(x) = 0.$$
We consider two subcases.\\
\textbf{Sub-case 3(a).}  $a \leq p^s-a+t_1$. In this case, Equation (4) gives
\[
 c_1(x) = f(x)^{p^s-a}s_1(x)-f(x)^{p^s-2a+t_1}h_1(x)\tilde{c}_0(x),
\]
for some $s_1(x)\in R^{1,\omega}$.  
Using this in Equation (5), we get
\[
\begin{aligned}
f(x)^a c_2(x) +f^{p^s-a+t_1}h_1(x)s_1(x)-f(x)^{p^s-2a+2t_1}h_1(x)^2\tilde{c}_0(x)=f(x)^{\gamma}.
\end{aligned}
\]
Thus $\gamma \geq \min\{a,\,p^s-2(a-t_1)\}$. In particular, since $u^3f(x)^{L}\in C$, we have $L \geq \min\{a,\,p^s-2(a-t_1)\}$.\\  Taking $\tilde{c}_0(x)=1,\,c_1(x)=-f(x)^{p^s-2a+t_1}h_1(x),$ and $c_2(x)=0$, we get $u^3f(x)^{p^s-2(a-t_1)} \in C$ and taking $\tilde{c}_0(x)=0,\,c_1(x)=0,$ and $c_2(x)=1$, we get $u^3f(x)^{a}\in C.$ Consequently, as in Case 2, $L=\min\{a,\,p^s-2(a-t_1)\}.$\\
\textbf{Sub-case 3(b).}  $a \geq p^s-a+t_1$. As observed at the beginning of this case $\gamma \geq t_1$. In particular, since $u^3f(x)^{L}\in C$, we have $L \geq t_1$. Also taking $c_0(x)=f(x)^{a-t_1}h_1(x)^{-2},\, c_1(x)=h_1(x)^{-1},\,c_2(x)=0$, we see that $\tilde{c}_0(x)=f(x)^{2a-t_1-p^s}h_1(x)^{-2}$ and  $u^3f(x)^{t_1}\in C$. Hence $L=t_1.$ \\
\textbf{Case 4.} $h_1(x)\neq 0$, $h_2(x)\not=0$.\\
As in Case 3, we consider two subcases.\\
\textbf{Sub-case 4(a).}  $a \geq p^s-a+t_1$.\\
In this case, Equation (4) gives
\[\tilde{c}_0(x)=h_1(x)^{-1}\{f(x)^{a-t_1}s_1(x)-f(x)^{2a-p^s-t_1}c_1(x)\]for some $s_1(x)\in R^{1,\omega}.$ Using this in Equation (5), we get \[f(x)^a+\{f^{t_1}h_1(x)-f(x)^{a+t_2-t_1}h_2(x)h_1(x)^{-1}\}c_1(x)+f(x)^{p^s+t_2-t_1}h_2(x)h_1(x)^{-1}s_1(x)=f(x)^{\gamma}.\] Thus $\gamma\geq \min\{a,\,\beta_2,\,p^s+t_2-t_1\},$ where $\beta_2=\text{max}\{k\,:\,f(x)^k\mid (f(x)^{t_1}h_1(x)-f(x)^{a+t_2-t_1}h_2(x)\\h_1^{-1}(x)\}$. In particular, since $u^3f(x)^{L}\in C$, we have $L \geq \min\{a,\,\beta_2,\,p^s+t_2-t_1\}$. For the other way, clearly $u^3f(x)^a\in C.$ Also taking $c_0(x)=-f(x)^{a-t_1}h_1(x)^{-1},\,c_1(x)=1,c_2(x)=0,$ we see that $u^3f(x)^{\beta_2}\in C,$ and taking $c_0(x)=f(x)^{p^s-t_1}h_2(x)^{-1},\,c_1(x)=c_2(x)=0$, we see that $u^3f(x)^{p^s+t_2-t_1}\in C.$ It follows that $L=\min\{a,\,\beta_2,\,p^s+t_2-t_1\}.$ \\
\textbf{Sub-case 4(b).} $a \leq p^s-a+t_1$. Similarly, one can prove that 
$L=\min\{a,\, p^s - a+t_1,\,\beta_1\},$ where $\beta_1=\text{max}\{k\,:\,f(x)^k\mid (f(x)^{p^s-a+t_2}h_2(x)-f(x)^{p^s-2a+2t_1}h_1^2(x)\}$.\hfill $\square$
\end{proof}
In the next two propositions, we give two more parameters without proof. The proofs of these propositions are similar to those of Proposition \ref{proofofresult} and Proposition \ref{Parameter 2}. Similar results hold for the other parameters that are mentioned in Theorem \ref{idealsofS'1}. Due to the complexity of the expressions, however, we are not including those here.  
\begin{proposition}
    Let $L$ be the smallest non-negative integer such that $u^3f(x)^{L}\in\langle u^2f(x)^{a_{1}}+u^3f(x)^{t_1}h_1(x),uf(x)^{a_2}+u^2f(x)^{t_2}h_2(x)+u^3f(x)^{t_3}h_3(x)\rangle,$ where for $1 \le i \le 3$,  $h_i(x)$, if non-zero, is a unit in $R^{1,\omega}.$ Then, 
    
\begin{adjustbox}{max width =15.56 cm}
$L= \begin{cases} 
      a_1 & \textnormal{ if } h_1(x)=h_2(x)=h_3(x)=0 \\

       \Min\{a_1,\,a_2-a_1+t_1\}  & \textnormal{ if } h_1(x)\not=0,\, h_2(x)=0, \,h_3(x)=0 \\

  t_2  & \textnormal{ if } h_1(x)=0,\,h_2(x)\not=0,\,h_3=0 \\
        \Min\{a_1,\,p^s-a_2+t_3\}  & \textnormal{ if } h_1(x)=0,\, h_2(x)=0, \,h_3(x)\not=0 \\
        
        \Min\{t_2,\,p^s-a_2+t_3\}  & \textnormal{ if } h_1(x)=0,\,h_2(x)\not=0,\,h_3(x)\not=0, \textnormal{ and } a_1\leq p^s-a_2+t_2\\
      
      \Min\{a_1,\,a_1+t_3-t_2,\,\beta_1\}  & \textnormal{ if } h_1(x)=0,\,h_2(x)\not=0,\,h_3(x)\not=0, \textnormal{ and } a_1\geq p^s-a_2+t_2\\

       \Min\{a_1,\,p^s-a_1-a_2+t_1+t_2,\,\beta_2\}  & \textnormal{ if } h_1(x)\not=0,\,h_2(x)\not=0,\,h_3(x)=0, \textnormal{ and } a_1\leq p^s-a_2+t_2\\
       
        \Min\{t_1,t_2\}  & \textnormal{ if } h_1(x)\not=0,\,h_2(x)\not=0,\,h_3(x)=0, \textnormal{ and } a_1\geq p^s-a_2+t_2\\

        \Min\{a_1,\,a_2-a_1+t_1,\,p^s-a_2+t_3\}  & \textnormal{ if } h_1(x)\not=0,\,h_2(x)=0,\,h_3(x)\not=0\\

        \Min\{a_1,\,p^s-a_1+t_1,\,\beta_3,\,\beta_4\}  & \textnormal{ if } h_1(x)\not=0,\,h_2(x)\not=0,\,h_3(x)\not=0, \textnormal{ and } a_1\leq p^s-a_2+t_2\\

        \Min\{a_1,\,p^s+t_3-t_2,\,\beta_5,\,\beta_6\}  & \textnormal{ if } h_1(x)\not=0,\,h_2(x)\not=0,\,h_3(x)\not=0, \textnormal{ and } a_1\geq p^s-a_2+t_2\\
   \end{cases}$
   \end{adjustbox}
   
   where $\beta_1=\textnormal{max}\{k\,:\,f^k\mid (f^{t_2}h_2-f^{a_2+t_3-t_2}h_3h_2^{-1})\},\,
   \beta_2=\textnormal{max}\{k\,:\, f^k\mid (f^{t_2}h_2-f^{a_2-a_1+t_1}h_1\},\,\\
   \beta_3=\textnormal{max}\{k\,:\, f^k\mid (f^{t_2}h_2-f^{a_2-a_1+t_1}h_1\}$, $\beta_4=\textnormal{max}\{k\,:\,f^k\mid(f^{p^s-a_2+t_3}h_3-f^{p^s-a_1-a_2+t_1+t_2}h_1\\h_2)\}$, $\beta_5=\textnormal{max}\{k\,:\,f^k\mid(f^{t_1}h_1+f^{a_1+t_3-t_2}h_3h_2^{-1})\}$, $\beta_6=\textnormal{max}\{k\,:\,f^k\mid(f^{t_2}h_2+f^{a_2+t_3-t_2}h_3h_2^{-1}\\)\}.$
\end{proposition}
\begin{proposition}
 Let $L$ be the smallest non-negative integer such that $u^3f(x)^{L}\in\langle f(x)^{b}+uf(x)^{t_1}h_1(x)+u^2f(x)^{t_2}h_2(x)+u^3f(x)^{t_3}h_3(x)\rangle,$ where for $1 \le i \le 3$,  $h_i(x)$, if non-zero, is a unit (if non-zero) in $R^{1,\omega}.$ Then, 
 
 \begin{adjustbox}{max width=15.56 cm}  
 $L= \begin{cases} 
      b & \textnormal{ if } h_1(x)=h_2(x)=h_3(x)=0 \\

      \Min\{b,\,p^s-3(b-t_1)\}  & \textnormal{ if } h_1(x)\not= 0,\, h_2(x)=0,\,h_3(x)=0,\,\textnormal{ and } b\leq p^s-(b-t_1),\,b\leq p^s-2(b-t_1) \\

      t_1  & \textnormal{ if }  h_1(x)\not =0,\, h_2(x)=0,\,h_3(x)=0,\,\textnormal{ and } b\leq p^s-(b-t_1),\,b\geq p^s-2(b-t_1)\\

      \Min\{b,\,p^s-b+t_2\}  & \textnormal{ if } h_2(x)\not= 0,\, h_1(x)=0,\,h_3(x)=0 \\

     \Min\{b,\,p^s-b+t_3\}  & \textnormal{ if } h_3(x)\not= 0,\, h_1(x)=0,\,h_2(x)=0 \\

    \Min\{b,\,p^s-2b+t_1+t_2,\,\beta_1\}  & \textnormal{ if } h_1(x)\not= 0,\, h_2(x)\not=0,\,h_3(x)=0,\,\textnormal{ and } b\leq p^s-(b-t_1),\,\\&f^{p^s-b+t_2}h_2=f^{p^s-2b+2t_1}h_1^2 \\

     \Min\{b,\,2p^s-2b+t_1+t_2-\beta_1,\beta_2,\,\beta_3\}  & \textnormal{ if } h_1(x)\not= 0,\, h_2(x)\not=0,\,h_3(x)=0,\,\textnormal{ and } b\leq p^s-(b-t_1),\,\\&f^{p^s-b+m_2}h_2\not=f^{p^s-2b+2t_1}h_1^2,\, \beta_1\leq b \\

    \Min\{b,\,\beta_4\}  & \textnormal{ if } h_1(x)\not= 0,\, h_2(x)\not=0,\,h_3(x)=0,\,\textnormal{ and } b\leq p^s-(b-t_1),\,\\&f^{p^s-b+t_2}h_2\not=f^{p^s-2b+2t_1}h_1^2,\, \beta_1\geq b \\

    t_2  & \textnormal{ if } h_1(x)\not= 0,\, h_2(x)\not=0,\,h_3(x)=0,\,\textnormal{ and } b\geq p^s-(b-t_1),\,\\&f^{t_1}h_1=f^{b+t_2-t_1}h_2h_1^{-1}\\

    \Min\{\beta_5,\,p^s+t_2-b\}  & \textnormal{ if } h_1(x)\not= 0,\, h_2(x)\not=0,\,h_3(x)=0,\,\textnormal{ and } b\geq p^s-(b-t_1),\,\\&f^{t_1}h_1\not=f^{b+t_2-t_1}h_2h_1^{-1},\, b\leq \alpha_1,\, b\leq p^s+t_2-t_1 \\

    \Min\{b,\,\beta_6,\,p^s+2t_2-t_1-\alpha_1\}  & \textnormal{ if } h_1(x)\not= 0,\, h_2(x)\not=0,\,h_3(x)=0,\,\textnormal{ and } b\geq p^s-(b-t_1),\,\\&f^{t_1}h_1\not=f^{b+t_2-t_1}h_2h_1^{-1},\, \alpha_1\leq b,\, \alpha_1\leq p^s+t_2-t_1 \\

    t_2  & \textnormal{ if } h_1(x)\not= 0,\, h_2(x)\not=0,\,h_3(x)=0,\,\textnormal{ and } b\geq p^s-(b-t_1),\,\\&f^{m_1}h_1\not=f^{b+t_2-t_1}h_2h_1^{-1},\, p^s+t_2-t_1\leq b,\, p^s+t_2-t_1\leq \alpha_1  \\

    \Min\{p^s-2(b-t_1),\beta_7\}  & \textnormal{ if } h_1(x)\not= 0,\, h_2(x)=0,\,h_3(x)\not=0,\,\textnormal{ and } b\leq p^s-(b-t_1),\,b\leq p^s-2(b-t_1)\\

    \Min\{b,\,p^s-b-2t_1+t_3,\,\beta_8\}  & \textnormal{ if } h_1(x)\not= 0,\, h_2(x)=0,\,h_3(x)\not=0,\,\textnormal{ and } b\leq p^s-(b-t_1),\,b\geq p^s-2(b-t_1)\\

    \Min\{b,\,p^s+t_3-t_1,\,\beta_{8}\}  & \textnormal{ if } h_1(x)\not= 0,\, h_2(x)=0,\,h_3(x)\not=0,\,\textnormal{ and } b\geq p^s-(b-t_1)\\

      \Min\{b,\,p^s-b+t_3\}  & \textnormal{ if } h_1(x)= 0,\, h_2(x)\not=0,\,h_3(x)\not=0,\,\textnormal{ and } b\leq p^s-(b-t_2)\\

      \Min\{b,\,p^s-b+t_2,\,b+t_3-t_2\}  & \textnormal{ if } h_1(x)= 0,\, h_2(x)\not=0,\,h_3(x)\not=0,\,\textnormal{ and } b\geq p^s-(b-t_2)\\

      \Min\{b,\,\beta_9,\,\beta_{10}\}  & \textnormal{ if } h_1(x)\not= 0,\, h_2(x)\not=0,\,h_3(x)\not=0,\,\textnormal{ and } b\leq p^s-(b-t_1),\,\\&f^{p^s-b+t_2}h_2=f^{p^s-2b+2t_1}h_1^2\\

      \Min\{p^s-b+t_1,\,\beta_{9},\,\beta_{11}\}  & \textnormal{ if } h_1(x)\not= 0,\, h_2(x)\not=0,\,h_3(x)\not=0,\,\textnormal{ and } b\leq p^s-(b-t_1),\,\\&f^{p^s-b+t_2}h_2\not=f^{p^s-2b+2t_1}h_1^2,\, b\leq \beta_9\\

      \Min\{b,\,p^s-2b+t_1+t_2,\,2p^s-b+t_3-\beta_9,\,\beta_{12},\,\beta_{13}\}  & \textnormal{ if } h_1(x)\not= 0,\, h_2(x)\not=0,\,h_3(x)\not=0,\,\textnormal{ and } b\leq p^s-(b-t_1),\,\\&f^{p^s-b+t_2}g_2\not=f^{p^s-2b+2t_1}h_1^2,\, b\geq \beta_9\\

      \Min\{p^s-b+t_1,\,\beta_{14},\,\beta_{15}\}  & \textnormal{ if } h_1(x)\not= 0,\,h_2(x)\not=0,\,h_3(x)\not=0,\,\textnormal{ and } b\geq p^s-(b-t_1),\,\\&f^{t_1}h_1=f^{b+t_2-t_1}h_2h_1^{-1},\, b\leq p^s+t_2-t_1\\

      \Min\{p^s+t_3-t_2,\,\beta_{15},\,\beta_{16}\}  & \textnormal{ if } h_1(x)\not= 0,\, h_2(x)\not=0,\,h_3(x)\not=0,\,\textnormal{ and } b\geq p^s-(b-t_1),\,\\&f^{t_1}h_1=f^{b_1+m_2-t_1}h_2h_1^{-1},\, b\geq p^s+t_2-t_1\\

       \Min\{p^s-b+t_1,\,\beta_{17},\,\beta_{18}\}  & \textnormal{ if } h_1(x)\not= 0,\, h_2(x)\not=0,\,h_3(x)\not=0,\,\textnormal{ and } b\geq p^s-(b-t_1),\,\\&f^{t_1}h_1\not=f^{b+t_2-t_1}h_2h_1^{-1},\, b\leq p^s+t_2-t_1,\,b\leq \alpha_2\,\\

       \Min\{b,\,\beta_{19},\,\beta_{20},\,\beta_{21}\}  & \textnormal{ if } h_1(x)\not= 0,\, h_2(x)\not=0,\,h_3(x)\not=0,\,\textnormal{ and } b\geq p^s-(b-t_1),\,\\&f^{t_1}h_1\not=f^{b_1+t_2-t_1}h_2h_1^{-1},\, \alpha_2\leq p^s+t_2-t_1,\,\alpha_2\leq b,\, f^{t_2}h_2\not=f^{b+t_3-t_1}h_3h_1^{-1}\\

      \Min\{b+t_3-t_2,\,\beta_{16},\,\beta_{22}\}  & \textnormal{ if } h_1(x)\not= 0,\, h_2(x)\not=0,\,h_3(x)\not=0,\,\textnormal{ and } b\geq p^s-(b-t_1),\,\\&f^{t_1}h_1\not=f^{b+t_2-t_1}h_2h_1^{-1},\, p^s+h_2-t_1\leq \alpha_2 ,\,p^s+t_2-t_1\leq b\\

      \end{cases}
         $
 \end{adjustbox}

      where $\beta_1=\textnormal{max}\{k\,:\,f^k\mid (f^{p^s-b+t_2}h_2-f^{p^s-2b+2t_1}h_1^{2})\}$, $ \beta_2=\textnormal{max}\{k\,:\,f^k\mid (f^{t_1}h_1+f^{p^s-b+t_1+t_2-\beta_1}h_2h_4^{-1})\}$, $ \beta_3=\textnormal{max}\{k\,:\,f^k\mid (f^{p^s-b+t_2}h_2+f^{2p^s-3b+2t_1+t_2-\beta_1}h_1h_2h_4^{-1})\}$, $ \beta_4=\textnormal{max}\{k\,:\,f^k\mid (f^{\beta_1-b+t_1}h_1h_4+f^{p^s-2b+t_1+t_2}h_2)\}$, $
      \alpha_1=\textnormal{max}\{k\,:\,f^k\mid (f^{t_1}h_1-f^{b+t_2-t_1}h_2h_1^{-1})\}
      $, $\beta_5=\textnormal{max}\{k\,:\,f^k\mid (f^{t_2}h_2-f^{t_1+\alpha_1-b}h_1h_5)\}$, $\beta_6=\textnormal{max}\{k\,:\,f^k\mid (f^{t_1}h_1-f^{b-\alpha_1+t_2}h_2h_5^{-1})\}$, $ \beta_7=\textnormal{max}\{k\,:\,f^k\mid (f^{p^s-3b+3t_1}h_1^3+f^{p^s-b+t_3}h_3)\}$, $\beta_8=\textnormal{max}\{k\,:\,f^k\mid (f^{t_1}h_1+f^{2b-2t_1+t_3}h_3h_1^{-1})\}$, $\beta_9=\textnormal{max}\{k\,:\,f^k\mid (f^{p^s-b+t_2}h_2-f^{p^s-2b+2t_1}h_1^{2})\}$, $\beta_{10}=\textnormal{max}\{k\,:\,f^k\mid (f^{p^s-b+t_3}h_3-f^{p^s-2b+t_1+t_2}h_1h_2)\}$, $\beta_{11}=\textnormal{max}\{k\,:\,f^k\mid (f^{p^s-b+t_3}h_3-f^{p^s-2b+t_1+t_2}h_1h_2-f^{\beta_9+t_1-b}h_1h_6)\}$, $ \beta_{12}=\textnormal{max}\{k\,:\,f^k\mid (f^{t_1}h_1-f^{p^s+t_3-\beta_9}h_3h_6^{-1})\}$, $ \beta_{13}=\textnormal{max}\{k\,:\,f^k\mid (f^{p^s-b+t_2}h_2-f^{2p^s-2b+t_1+t_3-\beta_9}h_1h_3h_6^{-1})\}$, $\beta_{14}=\textnormal{max}\{k\,:\,f^k\mid (f^{p^s+t_3-t_1}h_3h_1-f^{p^s+t_2-b}h_2)\}$, $ \beta_{15}=\textnormal{max}\{k\,:\,f^k\mid (f^{t_2}h_2-f^{b+t_3-t_1}h_3h_1^{-1})\}$, $\beta_{16}=\textnormal{max}\{k\,:\,f^k\mid (f^{t_1}h_1-f^{b+t_3-t_2}h_3h_2^{-1})\}$, $\alpha_{2}=\textnormal{max}\{k\,:\,f^k\mid (f^{t_1}h_1-f^{b+t_2-t_1}h_2h_1^{-1})\}$, $ \beta_{17}=\textnormal{max}\{k\,:\,f^k\mid (f^{t_2}h_2-f^{b+t_3-t_1}h_3h_1^{-1}-f^{\alpha_2-b+t_1}h_1h_7)\}$, $\beta_{18}=\textnormal{max}\{k\,:\,f^k\mid (f^{p^s+t_3-t_1}h_3h_1^{-1}-f^{p^s+t_2-b}h_2)\}$, $\beta_{19}=\textnormal{max}\{k\,:\,f^k\mid (f^{t_1}h_1-f^{b+t_2-\alpha_2}h_7^{-1}\\h_2)+f^{2b+t_3-t_1-\alpha_2}h_7^{-1}h_3h_1^{-1}\}$, $\beta_{20}=\textnormal{max}\{k\,:\,f^k\mid (f^{p^s+t_2-\alpha_2}h_7^{-1}h_2-f^{p^s-b+t_3-t_1-\alpha_2}h_7^{-1}h_3h_1^{-1})\\\}$, $\beta_{21}=\textnormal{max}\{k\,:\,f^k\mid (f^{p^s+b+t_3+t_2-2t_1}h_2^2h_1^{-2}h_7^{-1}-f^{p^s+2t_2-t_1}h_2^{2}h_1^{-1}h_7^{-1}+f^{p^s+t_3-t_1}h_3h_1^{-1})\}$, $\beta_{22}=\textnormal{max}\{k\,:\,f^k\mid (f^{t_2}h_2-f^{b+t_3-t_1}h_3h_1^{-1}+f^{\alpha_2+t_3-t_2}h_3h_2^{-1}h_7)\}$
      and $f^{p^s-b+t_2}h_2-f^{p^s-2b+2t_1}h_1^2\\=f^{\beta_1}h_4,\, f^{t_1}h_1-f^{b+t_2-t_1}h_2h_1^{-1}=f^{\alpha_1}h_5,\, f^{p^s-b+t_2}h_2-f^{p^s-2b+2t_1}h_1^2=f^{\beta_9}h_6,\, f^{t_1}h_1-f^{b+t_2-t_1}h_2\\h_1^{-1}=f^{\alpha_2}h_7$ with $h_4,h_5,h_6,h_7$ are units.
\end{proposition}

To obtain the number of codewords in the codes described in Theorem \ref{idealsofS'1}, we recall that for an ideal $C$ of $R^{t,\omega}$ and for $0\leq i \leq t-1$, the $i^{\textnormal{th}}$ torsion of $C$ is given by
$$\Tor_i(C)=\mu(\{c(x)\in R^{t,\omega}: c(x)u^i\in C\}).$$ 
Note that for $0\leq i \leq t-1$, $\Tor_i(C)$ is an ideal of $R^{1,\omega}$ and hence $\Tor_i(C)=\langle f(x)^{T_i}\rangle $ for some integer $T_i$ such that $0\leq T_i\leq p^s.$
\begin{lemma}\label{cardCisproductoftor}
     Let $f$ be an irreducible polynomial over $\mathbb{F}_{p^m}$, $C$ be a polycyclic code over $R^t$ associated with the polynomial $\omega(x)=f(x)^{p^s},$ where $s$ is a non-negative integer. Then $|C|=\underset{i=0}{\overset{t-1}{\prod}}|\Tor_i(C)|.$ 
\end{lemma}
\begin{proof}
We note that there exist natural numbers $k_0,\  k_1, \cdots, k_{t-1}$ such that generator matrix \( G \) of $C$ in \emph{standard form} is
\begin{equation*}
G =
\left(
\begin{array}{cccccc}
I_{k_0} & G_{0,1} & G_{0,2} & \cdots & G_{0,t-1} & G_{0,t} \\
0 & u I_{k_1} & u G_{1,2} & \cdots & u G_{1,t-1} & u G_{1,t} \\
\vdots & \vdots & \vdots & \ddots & \vdots & \vdots \\
0 & 0 & 0 & \cdots & u^{t-1} I_{k_{t-1}} & u^{t-1} G_{t-1,t}
\end{array}
\right)
U,
\end{equation*}

where \( I_{k_j} \) is an identity matrix of order \( k_j \)
(for \( 0 \le j \le t-1 \)) and \( U \) is a suitable permutation matrix (cf. \cite{norton2000structure}). To compute the cardinality of $C$, we can omit $U$. Thus
\begin{equation*}
    |C|=\prod_{j=0}^{t-1}|u^jR^t|^{k_j}
    =\prod_{j=0}^{t-1}(p^m)^{{(t-j)}k_j}
    =p^{m\sum_{j=0}^{t-1}(t-j)k_j}.
\end{equation*}
Also, then, the generator matrix $G_i$ for $\Tor_i(C)$ is:\\
\[ G_i=
\mu
\left(
\begin{array}{ccccc}
I_{k_0} & G_{0,1} & G_{0,2} & \cdots & G_{0,t} \\
0 & I_{k_1} & G_{1,2} & \cdots & G_{1,t} \\
\vdots & \vdots & \vdots & \ddots & \vdots \\
0 & 0 & \cdots & I_{k_i} & G_{i,t}
\end{array}
\right).
\]
Thus,
\begin{equation*}
    \prod_{i=0}^{t-1}|\Tor_i(C)|=\prod_{i=0}^{t-1}\prod_{j=0}^{i}p^{m{k_j}}
    =p^{m{\underset{i=0}{\overset{t-1}{\sum}}\underset{j=0}{\overset{i}{\sum}}k_j}}.
\end{equation*}
Hence
$|C|=\prod_{i=0}^{t-1}|\Tor_i(C)|.$\hfill $\square$
\end{proof}
The following result is critical for computing the cardinality of the codes. 
\begin{theorem}\label{card}
      Let $f(x)$ be an irreducible polynomial over $\mathbb{F}_{p^m}$, $\omega(x)=f(x)^{p^s}$ where $s$ be a non-negative integer, and let $C$ be a polycyclic code $R^t$ associated with polynomial $\omega(x)$, equivalently, an ideal $C$ of $R^{t,\omega},$
    \begin{itemize}
        \item [(i)] $|\Tor_i(C)|=(p^{md})^{p^s-T_i},$ for $0\leq i \leq t-1.$
        \item[(ii)] if for some  $g(x)\in R^{t,\omega}, \ u^i[f(x)^{l_i}+ug(x)]\in C,$ then $l_i\geq T_i.$ 
        \item[(iii)] $p^s\geq T_0 \geq T_1 \geq T_2 \geq \dots \geq T_{t-1} \geq 0.$
        \item[(iv)] $|C|=
        (p^{dm})^{4p^s-\underset{i=0}{\overset{t-1}{\sum}}T_i}.$
    \end{itemize}
\end{theorem}
\begin{proof}\begin{itemize}\item [(i)] As observed above, $\Tor_i(C)=\langle f(x)^{T_i} \rangle.$ Thus an arbitrary element $a(x)\in \Tor_i(C),$ can be written as $a(x)=f(x)^{T_i}b(x),$ for some $b(x)\in R^{1,\omega}.$ Hence, by simple calculations,
     \begin{equation*}
       a(x)=f(x)^{T_i}\underset{j=0}{\overset{p^s-1}{\sum}}\,\underset{k=0}{\overset{d-1}{\sum}}b_{k,j}x^kf(x)^j 
       =\underset{j=0}{\overset{p^s-1}{\sum}}\,\underset{k=0}{\overset{d-1}{\sum}}b_{k,j}x^kf(x)^{j+T_i}
       =\underset{j=0}{\overset{p^s-1-T_i}{\sum}}\,\underset{k=0}{\overset{d-1}{\sum}}b_{k,j}x^kf(x)^{j+T_i},
    \end{equation*}
    where $b_{k,j}\in \mathbb{F}_{p^m}.$ 
  Since $\{x^kf(x)^j\,|\,0\leq k\leq d-1,\,0\leq j\leq p^s-1\}$ forms an $\mathbb{F}_{p^m}$-basis of $R^{1,\omega}$ and every element of $\Tor_i(C)$ is uniquely written as linear combination of $\{x^kf(x)^j\,|\,0\leq k\leq d-1,\, 0\leq j\leq p^s-1-T_i\},$ it follows that $|\Tor_i(C)|=p^{md(p^s-T_i)}.$
  \item[(ii)] Follows from the definition of $\Tor_i(C)$, the $i^\textnormal{th}$ torsion of a code.
  \item [(iii)] Follows from the fact that  $\Tor_i(C)\subset \Tor_{i+1}(C)$.
  \item[(iv)]  Follows from (i) and Lemma \ref{cardCisproductoftor}.\hfill $\square$
  \end{itemize}
\end{proof}
For $0\leq i \leq t-1$, the integer $T_i$ in Theorem \ref{card} is called the $i$-th torsional degree of $C$ and is denoted by $T_i(C).$ In the following lemma, we obtain $T_i(C)$ for the ideals described in Theorem \ref{idealsofS'1}.

\begin{lemma}\label{Torsions}
     Let $f(x)$ be an irreducible polynomial over $\mathbb{F}_{p^m}$, $\omega(x)=f(x)^{p^s}$ where $s$ be a non-negative integer, and let $C$ be a polycyclic code over $R^4$ associated with the polynomial $\omega(x)$, equivalently, an ideal of $R^{4,\omega}$ as given in Theorem \ref{idealsofS'1}. Then
    \begin{enumerate}
        \item If $C=\langle 0 \rangle$, then $T_0(C)=T_1(C)=T_2(C)=T_3(C)=p^s.$
        \item If $C=\langle 1 \rangle,$ then $T_0(C)=T_1(C)=T_2(C)=T_3(C)=0.$
        \item If $C$ is as described in Theorem \ref{idealsofS'1} (\ref{Type 2}), then $T_0(C)=T_1(C)=T_2(C)=p^s,$ and $T_3(C)=a.$
        \item If $C$ is as described in Theorem \ref{idealsofS'1}.\ref{Type 3}, then $T_0(C)=T_1(C)=p^s,T_2(C)=a_2,$ and $T_3(C)=a_1.$
         \item If $C$ is as described in Theorem \ref{idealsofS'1} (\ref{Type 4}), then $T_0(C)=p^s$, $T_1(C)=a_2$, $T_2(C)=M,$ and $T_3(C)=a_1,$ where $M$ is as mentioned in Theorem \ref{idealsofS'1} (\ref{Type 4}).
         \item If $C$ is as described in Theorem \ref{idealsofS'1} (\ref{Type 5}), then $T_0(C)=p^s,\,T_1(C)=a_3,\, T_2(C)=a_2,$ and $T_3(C)=a_1.$ 
         \item If $C$ is as described in Theorem \ref{idealsofS'1} (\ref{Type 6}), then $T_0(C)=T_1(C)=p^s,\,T_2(C)=a,\,T_3(C)=L,$ where $L$ is as mentioned in Theorem \ref{idealsofS'1} (\ref{Type 6}).
         \item If $C$ is as described in Theorem \ref{idealsofS'1} (\ref{Type 7}), then $T_0(C)=p^s,\, T_1(C)=a_2,\,T_2(C)=a_1,\,T_3(C)=M,$ where $M$ is as mentioned in Theorem \ref{idealsofS'1} (\ref{Type 7}). 
         \item If $C$ is as described in Theorem \ref{idealsofS'1} (\ref{Type 8}), then $T_0(C)=p^s,\, T_1(C)=a_1,\,T_2(C)=L,\,T_3(C)=M,$ where $L$ and $M$ are as mentioned in Theorem \ref{idealsofS'1} (\ref{Type 8}).
         \item If $C$ is as described in Theorem \ref{idealsofS'1} (\ref{Type 9}), then $T_0(C)=b,\, T_1(C)=L,\,T_2(C)=M,\,T_3(C)=N,$ where $L,\,M,$ and $N$ are as mentioned in Theorem \ref{idealsofS'1} (\ref{Type 9}).
         \item If $C$ is as described in Theorem \ref{idealsofS'1} (\ref{Type 10}), then $T_0(C)=b,\, T_1(C)=M,\,T_2(C)=N,\,T_3(C)=a,$ where $M$ and $N$ are as mentioned in Theorem \ref{idealsofS'1} (\ref{Type 10}).
         \item If $C$ is as described in Theorem \ref{idealsofS'1} (\ref{Type 11}), then $T_0(C)=b,\, T_1(C)=N,\,T_2(C)=a_2,\,T_3(C)=a_1,$ where $N$ is as mentioned in Theorem \ref{idealsofS'1} (\ref{Type 11}).
         \item If $C$ is as described in Theorem \ref{idealsofS'1} (\ref{Type 12}), then $T_0(C)=b,\, T_1(C)=a_2,\,T_2(C)=N,\,T_3(C)=a_1,$ where $N$ is as mentioned in Theorem \ref{idealsofS'1} (\ref{Type 12}).
         \item If $C$ is as described in Theorem \ref{idealsofS'1} (\ref{Type 13}), then $T_0(C)=b,\, T_1(C)=a_3,\,T_2(C)=a_2,\,T_3(C)=a_1.$ 
         \item  If $C$ is as described in Theorem \ref{idealsofS'1} (\ref{Type 14}),  then $T_0(C)=b,\, T_1(C)=a_2,\,T_2(C)=a_1,\,T_3(C)=N,$ where $N$ is as mentioned in Theorem \ref{idealsofS'1} (\ref{Type 14}).
         \item If $C$ is as described in Theorem \ref{idealsofS'1} (\ref{Type 15}), then $T_0(C)=b,\, T_1(C)=M,\,T_2(C)=a,\,T_3(C)=N,$ where $M$ and $N$ are as mentioned in Theorem \ref{idealsofS'1} (\ref{Type 15}).
         \item If $C$ is as described in Theorem \ref{idealsofS'1} (\ref{Type 16}), then $T_0(C)=b,\, T_1(C)=a,\,T_2(C)=M,\,T_3(C)=N,$ where $M$ and $ N$ are as mentioned in Theorem \ref{idealsofS'1} (\ref{Type 16}).
    \end{enumerate}  
    \end{lemma}
    \begin{proof}
       If $C=\langle 0\rangle,$ then clearly $T_0(C)=T_1(C)=T_2(C)=T_3(C)=p^s.$ If $C=\langle 1\rangle,$ then clearly $T_0(C)=T_1(C)=T_2(C)=T_3(C)=0.$ If $C=\langle  u^3f(x)^{a}\rangle,$ then $T_0(C)=T_1(C)=T_2(C)=p^s$. By definition, $\Tor_3(C)=\mu\{c(x)\in R^{4,\omega} \,|\, c(x)u^3\in C\}.$ Note that $\mu(f(x)^{a})\in \Tor_3(C)\textnormal{ and hence } \langle f(x)^{a}\rangle\subset \Tor_3(C).$ Conversely, if $\mu(a(x))\in \Tor_3(C),$ for some $a(x)\in R^{4,\omega},$ then $a(x)u^3\in C\implies a(x)u^3=u^3f(x)^{a}h(x)$ for some $h(x)\in R^{4,\omega}.$ Then we have $u^3\underset{j=0}{\overset{p^s-1}{\sum}}\underset{i=1}{\overset{d-1}{\sum}}a_{i,j}^{(0)}x^if(x)^j=u^3f(x)^{a}\underset{j=0}{\overset{p^s-1}{\sum}}\underset{i=1}{\overset{d-1}{\sum}}h_{i,j}^{(0)}x^if(x)^j.$ Hence $\mu(a(x))\in \langle f(x)^{a}\rangle\implies \Tor_3(C)\subset \langle f(x)^{a}\rangle.$ 
        The procedure for calculating torsional degrees in other cases is similar.\hfill{$\square$}
    \end{proof}
Using Lemma \ref{Torsions} and Theorem \ref{card} (iv), we can now get the number of codewords in each of the polycyclic codes given in Theorem \ref{idealsofS'1}. We state this in our next theorem.
    \begin{theorem}\label{Cardinality}
         Let $f(x)$ be an irreducible polynomial over $\mathbb{F}_{p^m}$,  $\omega(x)=f(x)^{p^s}$ where $s$ be a non-negative integer, and let $C$ be a polycyclic code over $R^4$ associated with polynomial $\omega(x)$, equivalently, an ideal of $R^{4,\omega}$ as given in Theorem \ref{idealsofS'1}. Then we have the following:
        \begin{enumerate}
            \item If $C=\langle 0\rangle,$ then $|C|=1.$
            \item If $C=\langle1 \rangle,$ then $|C|=p^{4dmp^s}.$
            \item If $C$ is as described in Theorem \ref{idealsofS'1} (\ref{Type 2}), then $|C|=p^{dm(p^s-a)}.$  
            \item If $C$ is as described in Theorem \ref{idealsofS'1} (\ref{Type 3}), then $|C|=p^{dm(2p^s-a_1-a_2)}.$
            \item If $C$ is as described in Theorem \ref{idealsofS'1} (\ref{Type 4}), then $|C|=p^{dm(3p^s-a_1-a_2-M)},$ where $M$ is as mentioned in Theorem \ref{idealsofS'1} (\ref{Type 4}).
            \item If $C$ is as described in Theorem \ref{idealsofS'1} (\ref{Type 5}), then $|C|=p^{dm(3p^s-a_1-a_2-a_3)}.$
            \item If $C$ is as described in Theorem \ref{idealsofS'1} (\ref{Type 6}), then $|C|=p^{dm(3p^s-a_1-a_2-L)},$ where $L$ is as mentioned in Theorem \ref{idealsofS'1} (\ref{Type 6}).
            \item If $C$ is as described in Theorem \ref{idealsofS'1} (\ref{Type 7}), then $|C|=p^{dm(2p^s-a-M)},$ where $M$ is as mentioned in \ref{idealsofS'1} (\ref{Type 7}).
            \item If $C$ is as described in Theorem \ref{idealsofS'1} (\ref{Type 8}), then $|C|=p^{dm(3p^s-a_1-L-M)},$ where $L$ and $M$ are as mentioned in Theorem \ref{idealsofS'1} (\ref{Type 8}).
            \item If $C$ is as described in Theorem \ref{idealsofS'1} (\ref{Type 9}), then $|C|=p^{dm(4p^s-b-L-M-N)},$ where $L, \,M,$ and $N$ are as mentioned in Theorem \ref{idealsofS'1} (\ref{Type 9}).
            \item If $C$ is as described in Theorem \ref{idealsofS'1} (\ref{Type 10}), then $|C|=p^{dm(4p^s-b-M-N-a)},$ where $M$ and $N$ are as mentioned in Theorem \ref{idealsofS'1} (\ref{Type 10}).
            \item If $C$ is as described in Theorem \ref{idealsofS'1} (\ref{Type 11}), then $|C|=p^{dm(4p^s-b-a_1-a_2-N)},$ where $N$ is as mentioned in Theorem \ref{idealsofS'1} (\ref{Type 11}).
            \item If $C$ is as described in Theorem \ref{idealsofS'1} (\ref{Type 12}), then $|C|=p^{dm(4p^s-b-a_1-a_2-N)},$ where $N$ is as mentioned in Theorem \ref{idealsofS'1} (\ref{Type 12}).
            \item If $C$ is as described in Theorem \ref{idealsofS'1} (\ref{Type 13}), then $|C|=p^{dm(4p^s-b-a_1-a_2-a_3)}.$
            \item If $C$ is as described in Theorem \ref{idealsofS'1} (\ref{Type 14}), then $|C|=p^{dm(4p^s-b-a_1-a_{2}-N)},$ where $N$ is as mentioned in Theorem \ref{idealsofS'1} (\ref{Type 14}). 
            \item If $C$ is as described in Theorem \ref{idealsofS'1} (\ref{Type 15}), then $|C|=p^{dm(4p^s-b-a-M-N)},$ where $M$ and $N$ are as mentioned in Theorem \ref{idealsofS'1} (\ref{Type 15}). 
            \item If $C$ is as described in Theorem \ref{idealsofS'1} (\ref{Type 16}), then $|C|=p^{dm(4p^s-b-a-M-N)},$ where $M$ and $N$ are as mentioned in Theorem \ref{idealsofS'1} (\ref{Type 16}).
        \end{enumerate}
    \end{theorem}
 We conclude this section by computing the cardinality of some binary polycyclic codes.  
\begin{example}
     Let $p=2$, $t=4$, and $f(x)=x-1$. Then $f(x)$ is an irreducible polynomial over $\mathbb{F}_2$. Consider the ideal $\langle f(x)\rangle$ of the ring $\frac{\frac{\mathbb{F}_{2}[u]}{\langle u^4 \rangle}[x]}{\langle (x-1)^4 \rangle}.$ In this case, $T_0=T_1=T_2=T_3=1$. Hence the cardinality of $I=\langle f(x) \rangle$ is $2^{4\times 4-4}=4096.$ We verified this by MAGMA using the following code. 
    \begin{figure}[h]
\centering
\includegraphics[width=0.5\textwidth]{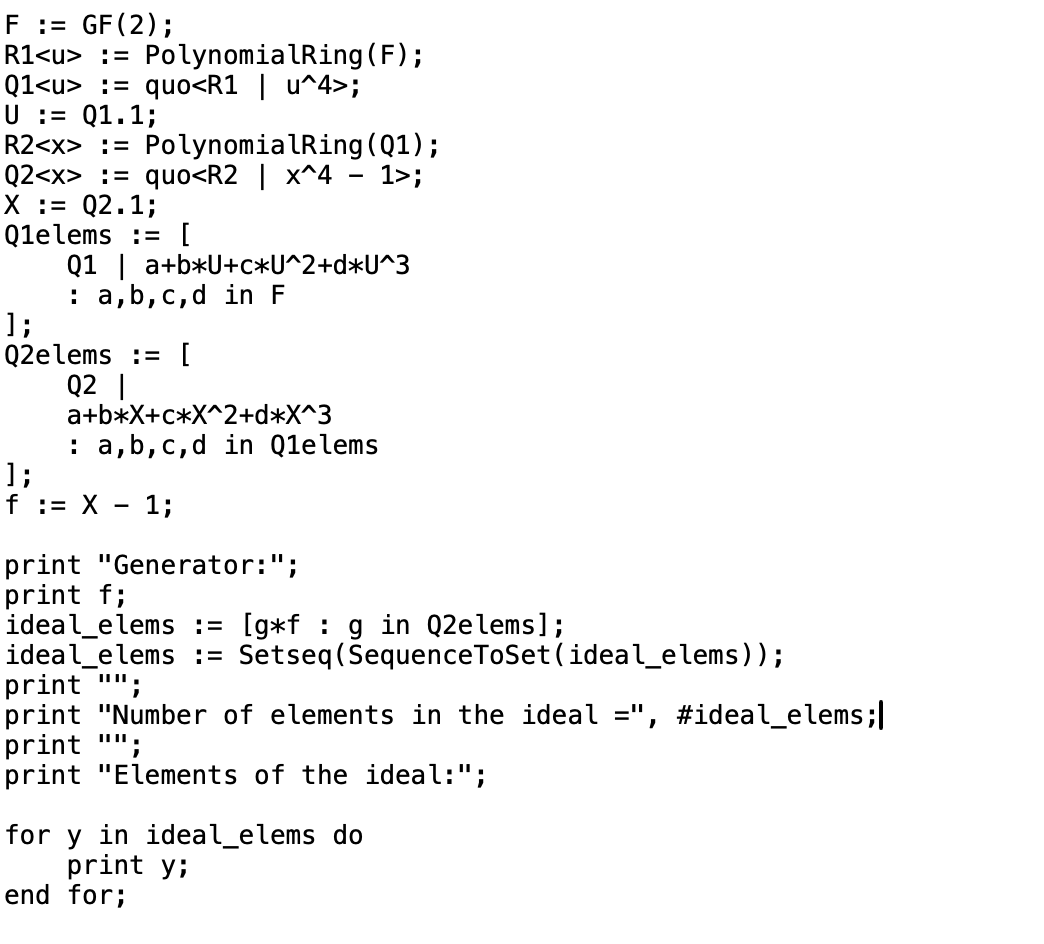}
\end{figure}
\end{example}
\begin{example}
      Let $p=2$, $t=4$, and $f(x)=x^2+x+1$. Then $f(x)$ is an irreducible polynomial over $\mathbb{F}_2$. Consider the ideal $\langle f(x)+u\rangle$ of the ring $\frac{\frac{\mathbb{F}_{2}[u]}{\langle u^4 \rangle}[x]}{\langle (x^2+x+1)^2 \rangle}.$ In this case, $T_0=T_1=1$ and $T_2=T_3=0.$ Hence the cardinality of $I=\langle x^2+x+1+u \rangle$ is $(2^2)^{4\times 2-2}=4096.$ It can be verified by MAGMA using a code similar to the code given in the above example. 
\end{example}

\section{Conclusion}
In this article, we first give a ring-theoretic result that helps us to get generators of an ideal of a ring whose image under a surjective ring homomorphism from the ring to another ring is finitely generated if the kernel of the homomorphism is principal. Using this result and techniques of basic commutative algebra, we obtain the ideals of the ring $R^{t,\,\omega}$ and their generators, extending the results for the case when $t=2$ and $\omega(x)= x^{p^s}-\lambda$ given in \cite{dinh2010constacyclic} and for the case when $t=3$ and $\omega(x)= x^{p^s}-1$ given in \cite{laaouine2021complete} to any value of $t$ and to any polynomial $\omega(x)$ over $\frac{\mathbb{F}_{p^m}[u]}{\langle u^t\rangle}$. In particular, for $\omega(x)=f(x)^{p^s},$ where $f(x)$ is an irreducible polynomial over $\mathbb{F}_{p^m},$ we find the ideals of $R^{t, \omega}.$
Furthermore, we compute, when ${t=4}$, certain parameters $L_i$'s for an irreducible polynomial $f(x)$ over $\mathbb{F}_{p^m}$ that help us in obtaining $i^\textnormal{th}$ torsion of codes for any irreducible polynomial $f(x)$ over $\mathbb{F}_{p^m}$. In Lemma \ref{cardCisproductoftor}, we give a relation between the cardinality of a polycyclic code (for any irreducible polynomial $f(x)$) and the cardinality of its torsions. Consequently, we compute cardinalities of these codes with the help of $i^\textnormal{th}$ torsional degree. For future direction, one can try to develop an efficient way or algorithm to compute $L_i$'s, since even for the case $t=4$, the computations become tedious and challenging.\\

\textbf{Conflict of Interest.} All authors declare that they have no conflict of interest.
\section*{Acknowledgments}
The first author would like to acknowledge PMRF (PMRF Id: 1403187) for its financial support.

\nocite{dinh2024constacyclic}
\bibliographystyle{abbrv}
\bibliography{sample}

@article{castagnoli1991repeated,
  title={On repeated-root cyclic codes},
  author={Castagnoli, Guy and Massey, James L and Schoeller, Philipp A and Von Seemann, Niklaus},
  journal={IEEE Transactions on Information Theory},
  volume={37},
  number={2},
  pages={337--342},
  year={1991},
  publisher={IEEE}
}

@article{van1991repeated,
  title={Repeated-root cyclic codes},
  author={van Lint, Jacobus H},
  journal={IEEE Transactions on Information Theory},
  volume={37},
  number={2},
  pages={343--345},
  year={1991},
  publisher={IEEE}
}

@article{hammons1994z,
  title={The $\mathbb{Z}_4$-linearity of {K}erdock, {P}reparata, {G}oethals, and related codes},
  author={Hammons, A Roger and Kumar, P Vijay and Calderbank, A Robert and Sloane, Neil JA and Sol{\'e}, Patrick},
  journal={IEEE Transactions on Information Theory},
  volume={40},
  number={2},
  pages={301--319},
  year={1994},
  publisher={IEEE}
}

@article{dinh2005negacyclic,
  title={Negacyclic codes of length $2^s$ over {G}alois rings},
  author={Dinh, Hai Q},
  journal={IEEE Transactions on Information Theory},
  volume={51},
  number={12},
  pages={4252--4262},
  year={2005},
  publisher={IEEE}
}

@article{wolfman1999negacyclic,
  title={Negacyclic and cyclic codes over $\mathbb{Z}_ 4$},
  author={Wolfman, J},
  journal={IEEE Transactions on Information Theory},
  volume={45},
  number={7},
  pages={2527--2532},
  year={1999},
  publisher={IEEE}
}

@article{calderbank1995modular,
  title={Modular and $p$-adic cyclic codes},
  author={Calderbank, A Robert and Sloane, Neil JA},
  journal={Designs, Codes and Cryptography},
  volume={6},
  pages={21--35},
  year={1995},
  publisher={Springer}
}

@article{kanwar1997cyclic,
  title={Cyclic codes over the integers modulo $p^m$},
  author={Kanwar, Pramod and Lopez-Permouth, Sergio R},
  journal={Finite Fields and Their Applications},
  volume={3},
  number={4},
  pages={334--352},
  year={1997},
  publisher={Elsevier}
}

@article{norton2000structure,
  title={On the structure of linear and cyclic codes over a finite chain ring},
  author={Norton, Graham H and S{\u{a}}l{\u{a}}gean, Ana},
  journal={Applicable Algebra in Engineering, Communication and Computing},
  volume={10},
  pages={489--506},
  year={2000},
  publisher={Springer}
}

@article{dinh2004cyclic,
  title={Cyclic and negacyclic codes over finite chain rings},
  author={Dinh, Hai Quang and L{\'o}pez-Permouth, Sergio R},
  journal={IEEE Transactions on Information Theory},
  volume={50},
  number={8},
  pages={1728--1744},
  year={2004},
  publisher={IEEE}
}

@article {MR1809649,
    AUTHOR = {Wan, Zhe-Xian},
     TITLE = {Cyclic codes over {G}alois rings},
   JOURNAL = {Algebra Colloquium},
  FJOURNAL = {Algebra Colloquium},
    VOLUME = {6},
      YEAR = {1999},
    NUMBER = {3},
     PAGES = {291--304},
      ISSN = {1005-3867,0219-1733},
   MRCLASS = {94B15},
  MRNUMBER = {1809649},
MRREVIEWER = {T.\ Aaron\ Gulliver},
}

@article{taher2003generators,
  title={On the Generators of $\mathbb{Z}_4$ Cyclic Codes of Length $2^e$},
  author={Abualrub, Taher and Oehmke, Robert},
  journal={IEEE Transactions on Information Theory},
  volume={49},
  pages={2126--2133},
  year={2003}
}

@ARTICLE{Dinhdist,
  author={Dinh, Hai Q.},
  journal={IEEE Transactions on Information Theory}, 
  title={Complete Distances of All Negacyclic Codes of Length  $2^{s}$ Over $\mathbb{Z} _{2^{a}}$}, 
  year={2007},
  volume={53},
  number={1},
  pages={147-161},
  doi={10.1109/TIT.2006.887487}}

@article{abualrub2004mass,
  title={A mass formula and rank of $\mathbb{Z}_4$ cyclic codes of length $2^e$},
  author={Abualrub, Taher and Ghrayeb, Ali and Oehmke, R},
  journal={IEEE Transactions on Information Theory},
  volume={50},
  number={12},
  pages={3307},
  year={2004}
}

@article{dinh2008linear,
  title={On the linear ordering of some classes of negacyclic and cyclic codes and their distance distributions},
  author={Dinh, Hai Q},
  journal={Finite Fields and Their Applications},
  volume={14},
  number={1},
  pages={22--40},
  year={2008},
  publisher={Elsevier}
}

@article{dinh2009constacyclic,
  title={Constacyclic Codes of Length $ 2^s $ Over {G}alois Extension Rings of $\mathbb{F}_2 +u \mathbb{F}_2 $},
  author={Dinh, Hai Q},
  journal={IEEE Transactions on Information Theory},
  volume={55},
  number={4},
  pages={1730--1740},
  year={2009},
  publisher={IEEE}
}

@article{dinh2010constacyclic,
  title={Constacyclic codes of length $p^s$ over $\mathbb{F}_{p^m}+u\mathbb{F}_{p^m}$},
  author={Dinh, Hai Q},
  journal={Journal of Algebra},
  volume={324},
  number={5},
  pages={940--950},
  year={2010},
  publisher={Elsevier}
}

@article{dinh2012repeated,
  title={Repeated-root constacyclic codes of length $2p^s$},
  author={Dinh, Hai Q},
  journal={Finite Fields and Their Applications},
  volume={18},
  number={1},
  pages={133--143},
  year={2012},
  publisher={Elsevier}
}

@article{chen2016constacyclic,
  title={Constacyclic codes of length $2p^s$ over $\mathbb{F}_{p^m}+ u\mathbb{F}_{p^m}$},
  author={Chen, Bocong and Dinh, Hai Q and Liu, Hongwei and Wang, Liqi},
  journal={Finite Fields and Their Applications},
  volume={37},
  pages={108--130},
  year={2016},
  publisher={Elsevier}
}

@article{liu2016repeated,
  title={Repeated-root constacyclic codes of length $3lp^s$ and their dual codes},
  author={Liu, Li and Li, Lanqiang and Kai, Xiaoshan and Zhu, Shixin},
  journal={Finite Fields and Their Applications},
  volume={42},
  pages={269--295},
  year={2016},
  publisher={Elsevier}
}

@article{dinh2024constacyclic,
  title={On constacyclic codes of length $9p^s$ over $\mathbb{F}_{p^m}$ and their optimal codes},
  author={Dinh, Hai Q and Ha, Hieu V and Nguyen, Nhan TV and Tran, Nghia TH},
  journal={Journal of Algebra and Its Applications},
  volume={23},
  number={08},
  pages={2550076},
  year={2024},
  publisher={World Scientific}
}

@article{laaouine2021complete,
  title={Complete classification of repeated-root $\sigma$-constacyclic codes of prime power length over $\mathbb{F}_{p^m} [u]/<u^3>$},
  author={Laaouine, Jamal and Charkani, Mohammed Elhassani and Wang, Liqi},
  journal={Discrete Mathematics, Algorithms and Applications},
  volume={344},
  number={6},
  pages={112325},
  year={2021},
  publisher={Elsevier}
}

@article{boudine2023classification,
  title={On the classification of ideals over ${R} [X]/\langle f (X)^{ p^s}\rangle $ when ${R}= \mathbb{F}_{p^m}+ u \mathbb{F}_{p^m}+…+ u^n \mathbb{F}_{p^m}$},
  author={Boudine, Brahim and Laaouine, Jamal and Charkani, Mohammed Elhassani},
  journal={Cryptography and Communications},
  volume={15},
  number={3},
  pages={589--598},
  year={2023},
  publisher={Springer}
}

@article{yildiz2010linear,
  title={Linear codes over $\mathbb{F}_2+u \mathbb{F}_2+v\mathbb{F}_2+uv\mathbb{F}_2$},
  author={Yildiz, Bahattin and Karadeniz, Suat},
  journal={Designs, Codes and Cryptography},
  volume={54},
  number={1},
  pages={61--81},
  year={2010},
  publisher={Springer}
}

@article{honold2000linear,
  title={Linear codes over finite chain rings},
  author={Honold, Thomas and Landjev, Ivan},
  journal={The Electronic Journal of Combinatorics},
  volume={7},
  pages={R11--R11},
  year={2000}
}

@article{dertli2016linear,
  title={On the linear codes over the ring ${R}_p$},
  author={Dertli, Abdullah and Cengellenmis, Yasemin and Eren, Senol},
  journal={Discrete Mathematics, Algorithms and Applications},
  volume={8},
  number={02},
  pages={1650036},
  year={2016},
  publisher={World Scientific}
}

@article{yildiz2007weights,
  title={Weights modulo $p^e$ of linear codes over rings},
  author={Yildiz, Bahattin},
  journal={Designs, Codes and Cryptography},
  volume={43},
  number={2},
  pages={147--165},
  year={2007},
  publisher={Springer}
}

@article{AlexandreFotue-Tabue2020AdvancesinMathematicsofCommunications,
title = {On polycyclic codes over a finite chain ring},
journal = {Advances in Mathematics of Communications},
volume = {14},
number = {3},
pages = {455-466},
year = {2020},
issn = {1930-5346},
doi = {10.3934/amc.2020028},
url = {https://www.aimsciences.org/article/id/be3f48b3-591d-4074-adcb-672d903f8f08},
author = {Alexandre Fotue-Tabue and Edgar Martínez-Moro and J. Thomas Blackford}
}

@article{hesari2024torsion,
  title={Torsion codes of a $\alpha$-constacyclic code over $\mathbb{F}_{p^m}[u]$},
  author={Hesari, Roghaye Mohammadi and Samei, Karim},
  journal={Discrete Mathematics},
  volume={347},
  number={6},
  pages={113919},
  year={2024},
  publisher={Elsevier}
}

@article{lopez2013polycyclic,
  title={Polycyclic codes over {G}alois rings with applications to repeated-root constacyclic codes},
  author={L{\'o}pez-Permouth, Sergio R and {\"O}zadam, Hakan and {\"O}zbudak, Ferruh and Szabo, Steve},
  journal={Finite Fields and Their Applications},
  volume={19},
  number={1},
  pages={16--38},
  year={2013},
  publisher={Elsevier}
}

@article{lopez2009dual,
  title={Dual generalizations of the concept of cyclicity of codes.},
  author={L{\'o}pez-Permouth, Sergio R and Parra-Avila, Benigno R and Szabo, Steve},
  journal={Advances in Mathematics of Communications},
  volume={3},
  number={3},
  pages={227--234},
  year={2009}
}

@article{liu2014some,
  title={Some classes of repeated-root constacyclic codes over $\mathbb{F}_{p^m}+u\mathbb{F}_{p^m}+u^2\mathbb{F}_{p^m}$},
  author={Liu, Xiusheng and Xu, Xiaofang},
  journal={Journal of the Korean Mathematical Society},
  volume={51},
  number={4},
  pages={853--866},
  year={2014}
}

\end{document}